\documentclass[11pt]{article}

\usepackage[utf8]{inputenc}
\usepackage[T1]{fontenc}

\usepackage[a4paper,margin=1in]{geometry}

\usepackage{comment}
\usepackage{xcolor}
\usepackage{hyperref}

\hypersetup{%
  colorlinks=true,
  linkbordercolor=red,
  citecolor=green!70!black,
}

\usepackage{authblk} 
\usepackage{titlesec} 
 
\usepackage{amsthm, amsmath, amssymb}
\usepackage{float}

\usepackage{enumitem}
\usepackage{needspace}
\usepackage{thmtools}
\usepackage{thm-restate}

\usepackage[linesnumbered,algoruled,boxed,lined]{algorithm2e}
\usepackage{complexity} 

\usepackage{needspace}
\usepackage{graphicx}
\usepackage{cleveref} 
\usepackage{soul}

\usepackage{tikz}
\usetikzlibrary{positioning,backgrounds,patterns,calc}

\usetikzlibrary{patterns}

\usepackage{subcaption}
\newcommand{\SoPP}{\textsc{Isometric Path Partition}\xspace}

\newcommand{\setrep}{\texttt{set-rep}}

\newcommand{\ETH}{{\textsf{ETH}}}

\newcommand{\diam}{\textsf{diam}}
\newcommand{\tw}{\textsf{tw}}

\newcommand{\pw}{\textsf{pw}}

\newcommand{\calP}{\mathcal{P}}
\newcommand{\calQ}{\mathcal{Q}}
\newcommand{\calS}{\mathcal{S}}

\newcommand{\true}{\texttt{True}}
\newcommand{\false}{\texttt{False}}

\newcommand{\yes}{\textsc{Yes}}

\def\th{\text{th}} 

\makeatletter
\newcommand{\customlabel}[2]{%
\protected@write \@auxout {}{\string \newlabel {#1}{{#2}{}}}}
\makeatother

\newcommand{\defproblem}[3]{
  \vspace{1mm}
\noindent\fbox{
  \begin{minipage}{0.96\textwidth}
  \begin{tabular*}{\textwidth}{@{\extracolsep{\fill}}lr} #1 \\ \end{tabular*}
  {\bf{Input:}} #2  \\
  {\bf{Question:}} #3
  \end{minipage}
  }
  \vspace{1mm}
}

\newcommand{\defproblemout}[3]{
  \vspace{1mm}
\noindent\fbox{
  \begin{minipage}{0.96\textwidth}
  \begin{tabular*}{\textwidth}{@{\extracolsep{\fill}}lr} #1 \\ \end{tabular*}
  {\bf{Input:}} #2  \\
  {\bf{Output:}} #3
  \end{minipage}
  }
  \vspace{1mm}
}

\newenvironment{claimproof}[1]{\begin{proof}#1}{\end{proof}}

\def\cqedsymbol{\ifmmode$\lrcorner$\else{\unskip\nobreak\hfil
\penalty50\hskip1em\null\nobreak\hfil$\lrcorner$
\parfillskip=0pt\finalhyphendemerits=0\endgraf}\fi} 
\newcommand{\claimqedhere}{\renewcommand{\qed}{\cqedsymbol}}

\newtheorem{theorem}{Theorem}[section]
\newtheorem{lemma}[theorem]{Lemma}
\newtheorem{corollary}[theorem]{Corollary}

\newtheorem{definition}[theorem]{Definition}
\newtheorem{claim}[theorem]{Claim}
\newtheorem{proposition}[theorem]{Proposition}
\newtheorem{reduction rule}[theorem]{Reduction Rule}
\newtheorem{marking-scheme}{Marking Scheme}

\usepackage[textsize=small, color=black!50, backgroundcolor=brown!50!white]{todonotes}

\title{Parameterized Complexity of Isometric\\ Path Partition: Treewidth and Diameter}

\author[1]{Dibyayan Chakraborty} 
\author[2]{Oscar Defrain} 
\author[3]{Florent Foucaud} 

\author[4]{\authorcr Mathieu Mari} 
\author[5]{Prafullkumar Tale} 

\affil[1]{School of Computer Science, University of Leeds, UK.}
\affil[2]{Aix-Marseille Université, CNRS, LIS, Marseille, France.}
\affil[3]{Université Clermont Auvergne, CNRS, Mines Saint-Étienne,\authorcr Clermont Auvergne INP, LIMOS, 63000 Clermont-Ferrand, France}
\affil[4]{LIRMM, Université de Montpellier, CNRS, Montpellier, France.}
\affil[5]{Indian Institute of Science Education and Research Pune, Pune, India}

\date{August 2025} 

\begin{document}

\maketitle

\begin{abstract}

In the \SoPP problem, the input is a graph $G$ with $n$ vertices
and an integer $k$, and the objective is to determine
whether the vertices of $G$ can be partitioned
into $k$ vertex-disjoint shortest paths.
We investigate the parameterized complexity of the problem
when parameterized by the treewidth ($\tw$)
of the input graph, arguably one of the most widely studied
parameters.
Courcelle's theorem [Information \& Computation, $1990$] shows
that graph problems that are expressible as MSO formulas
of constant size admit \FPT\ algorithms parameterized by the treewidth
of the input graph.
This encompasses many natural graph problems.
However, many metric-based graph problems,
where the solution is defined using some metric-based property of the graph
(often the distance) are not expressible as MSO formulas of constant size.
These types of problems, \SoPP being one of them, require individual attention and often
draw the boundary for the success story of parameterization by treewidth.

We prove that \SoPP is $\W[1]$-hard when parameterized by treewidth (in fact, even pathwidth ($\pw$)), answering the question by Dumas et al.~[SIDMA, 2024],
Fernau et al.~[CIAC, 2023], and confirming the aforementioned tendency.
We complement this hardness result by
designing a tailored dynamic programming algorithm running in $n^{O(\tw)}$ time.
This dynamic programming approach also results in
an algorithm running in time $\smash{\diam^{O(\tw^2)} \cdot n^{O(1)}}$, where $\diam$ is the diameter of the graph.
It is known that \SoPP remains \NP-hard on graphs of diameter~$2$; hence, the combination of both parameters is necessary to obtain a tractable algorithm.
Note that the dependency on treewidth is unusually high, as most problems admit algorithms running in time $2^{O(\tw)}\cdot n^{O(1)}$ or
$2^{O(\tw \log (\tw))}\cdot n^{O(1)}$.
However, we rule out the possibility of a significantly faster algorithm by
proving that \SoPP
does not admit an algorithm running in time $\smash{\diam^{o(\pw^2/(\log^3(\pw)))} \cdot n^{O(1)}}$, unless the {\textsf{Randomized-ETH}} fails.

\vskip5pt\noindent{}{\bf Keywords:} \SoPP, parameterized complexity, parameterized reductions, treewidth, diameter, Randomized ETH.
\end{abstract}

\newpage
\tableofcontents

\newcommand{\myparagraph}[1]{\paragraph{#1}}

\section{Introduction}\label{sec:intro}

In this paper, we investigate the parameterized complexity of a metric-based optimization problem known as \SoPP, that deals with partitioning the vertex set of an input graph into a given number of isometric (i.e., shortest) paths.

\myparagraph{Metric-based optimization problems.} The main subject of metric graph theory is the investigation and
characterization of graph classes and graph problems, where the graphs are equipped with a metric~\cite{DBLP:journals/tcs/ChalopinCCJ24,bandelt2008metric}.
It is a central topic in mathematics and computer science
with far-reaching applications such as in group theory~\cite{gromov1987,agol2013virtual},
matroid theory \cite{DBLP:journals/jct/BandeltCK18},
learning theory \cite{DBLP:journals/siamdm/ChalopinCIRV23,DBLP:journals/jcss/ChalopinCMW22,DBLP:journals/siamdm/ChepoiKP22,DBLP:conf/colt/ChalopinCIR24}, and
computational biology~\cite{bandelt1992split}.
One of the most natural metrics related to graphs
is the shortest-path distance between two vertices.
On the algorithmic side, many problems related to network monitoring,
transportation networks, information retrieval, or
computational learning can often be formulated as
problems on graphs in which the objective is to find vertices
that satisfy specified distance-related properties.
We use ``metric-based optimization problems'' as an umbrella 
term for such problems. 
This includes many important and classic graph problems, such as
\textsc{Single Source Shortest Paths},
\textsc{{Distance $d$}-Dominating Set} (also called \textsc{$(k,d)$-Center}),
\textsc{{Distance} $d$-Independent Set} (also called \textsc{$d$-Scattered Set}),
\textsc{Metric Dimension}, \textsc{Geodetic Set}, \textsc{Isometric Path Cover}, etc.

Some of these problems have been cornerstones in the development of
classic as well as parameterized algorithms and complexity
\cite{BelmonteFGR17,KLP19,LM21,KLP22,chakraborty2025distance,DBLP:conf/icalp/FoucaudGK0IST24,BDM23}, as they behave quite differently from their more ``local'' (neighborhood-based) counterparts such as \textsc{Vertex Cover}, \textsc{Independent Set} or \textsc{Dominating Set}.
{In} parameterized analysis, we associate each instance
$I$ with a parameter $\ell$, and are interested in an algorithm
with running time $f(\ell) \cdot |I|^{O(1)}$
for some computable function $f$.
Parameterized problems that admit such an algorithm are called
\emph{fixed parameter tractable} (\FPT) parameterized by $\ell$.
On the other hand, $\W[1]$-hardness categorizes problems that
are unlikely to have \FPT\ algorithms. A parameterized problem is in \XP~if it admits an algorithm running in time $|I|^{f(\ell)}$ for some computable function~$f$.

\myparagraph{Limitations of Treewidth.} A large class of problems admits \FPT\ algorithms when parameterized by
the treewidth, a parameter that quantifies
tree-likeness of the graph.
Courcelle's celebrated theorem~\cite{DBLP:journals/iandc/Courcelle90}
states that the class of graph problems expressible in Monadic Second-Order (MSO) Logic of constant size is fixed-parameter tractable (\FPT) when
parameterized by the treewidth of the graph.
We refer readers to \cite[Chapter 7]{cygan2015parameterized}
for further details.

Although one can express many of the graph properties using MSO
formulas of constant size,
there is no such formula to encode the following: given a subset of vertices and two specified vertices $s,t$, 
does this subset form an isometric path (i.e., a shortest path) between $s$ and $t$ \cite{denis2023stack}. 
This hinders the application of Courcelle's theorem to metric-based optimization problems.

\sloppy Consider the example of \textsc{Dominating Set} and its generalization
\textsc{{Distance $d$}-Dominating Set}.
The objectives of these problems are to find a subset of vertices $S$
such that any vertex in $V(G)\setminus S$ is at distance at most~1
and at most~$d$, respectively, from at least one vertex in $S$.
As the 
distance requirement in the first problem is upper bounded by a constant, it is expressible as an MSO formula of constant size,
resulting in an \FPT\ algorithm parameterized by treewidth.
However, this is not the case for the latter problem.
In fact, if $d$ is part of the input, it is known that \textsc{Distance $d$-Dominating Set} is \W[1]-hard when parameterized
by treewidth~\cite{DBLP:conf/iwpec/BorradaileL16}.
There are similar results for \textsc{Independent Set} and its
generalization \textsc{Distance $d$-Independent Set}~\cite{KLP22}.
 Similarly, the metric-based optimization problems \textsc{Geodetic Set} and \textsc{Metric Dimension} are even \NP-hard when
the treewidth of the graph is a constant~\cite{LM21,DBLP:journals/corr/abs-2504-17862}.
Hence, metric-based optimization problems require individual
attention and often draw the boundary for the success story of
parameterization by treewidth stemming from Courcelle's theorem. 

\myparagraph{\SoPP.} Isometric (i.e., shortest) paths in graphs and vertex-partitioning are among the most fundamental constructs in the area of graph algorithms. In this article, we consider an interesting metric-based optimization problem known as \SoPP, whose objective is to partition the vertex set of a graph into a prescribed number of isometric paths. Formally it is defined as follows.

\medskip

\defproblemout{\SoPP}
{A graph $G$ and an integer $k$.}{Is there a partition of the vertex set of $G$ into $k$ sets, each of them forming an isometric path in $G$?}

\medskip

Algorithmic aspects of \SoPP received increasing attention in recent years~\cite{pmanuelisometric,DBLP:conf/mfcs/ChakrabortyCFV23,PPP}. (We discuss the related literature in detail later.)
It is also related to other (non-metric based) path problems such as the celebrated \textsc{Hamiltonian Path} (and its generalization \textsc{Path Partition}) or \textsc{Disjoint Paths}, which are fundamental and have numerous applications~\cite{DBLP:journals/dam/AndreattaM95,manuel2018revisiting,DBLP:journals/jct/RobertsonS95b}.
  
\myparagraph{Our results.} As our first result, we show that the problem is \XP\ parameterized
by treewidth.

\begin{restatable}{theorem}{restatethmIPPdpxptw}
    \label{thm:IPP:DP-xp-tw}
    \SoPP admits an algorithm running in time $n^{O(\tw)}$, where $\tw$ is the treewidth of $G$ and $n$ denotes its number of vertices.
\end{restatable}

We note that \Cref{thm:IPP:DP-xp-tw} improves upon results
from~\cite{dumas2024graphs} and \cite{PPP}. 
Indeed, the authors from~\cite{dumas2024graphs}
showed that in
a \yes-instance, the pathwidth (and thus treewidth) is upper-bounded by
an exponential function of the solution size of \SoPP. They used this
fact combined with Courcelle's theorem to obtain an \XP\ algorithm
for the parameter solution size.
A different method is used in~\cite{PPP} to obtain another 
\XP\ algorithm for solution size. 
Using the aforementioned upper bound from~\cite{dumas2024graphs}, \Cref{thm:IPP:DP-xp-tw}
implies these results.

The next natural question is whether the above \XP\ algorithm can be
improved to an \FPT\ algorithm?
Recall that closely related ``path problems'' like
\textsc{Hamiltonian Path} and \textsc{Path Partition} are both
\FPT\ parameterized by treewidth~\cite{DBLP:journals/talg/CyganNPPRW22,PCcaldam}. 
We show that this is unlikely to be the case for \SoPP.

\begin{restatable}{theorem}{restateIPPtwhard}
\label{thm:IPP:w1-tw-hard}
    \SoPP is \emph{\W[1]-hard} when parameterized by the pathwidth, and hence, the treewidth of the input graph.
\end{restatable}

\Cref{thm:IPP:w1-tw-hard} answers open questions
from~\cite{dumas2024graphs} and~\cite{PPP}. 
Moreover, \Cref{thm:IPP:DP-xp-tw} and \Cref{thm:IPP:w1-tw-hard}
establish that \SoPP belongs to the list of problems that are \XP\ but
\W[1]-hard for treewidth. We refer
to~\cite{DBLP:journals/siamdm/BelmonteKLMO22} for a discussion about
such problems. It appears that certain common problem features
yielding this behavior can be listed, for example, problems involving
weights, lists, or iterative processes. Another kind of such feature
is the fact of being metric-based, such as \textsc{Metric
Dimension}~\cite{LM21}, \textsc{Geodetic Set}~\cite{KK22} and
\textsc{Distance-$d$ Dominating/Independent Set}~\cite{KLP22,KLP19}. Our
result confirms this trend and draws an interesting distinction with
the related (path-based but not metric-based) \textsc{Path Partition},
which is \FPT\ for treewidth~\cite{PCcaldam}.

For metric-based problems, another relevant parameter is the diameter
of the graph, 
which is the maximum length of an isometric path. 
Unfortunately, \SoPP is \NP-hard even on (chordal) graphs of
diameter~2~\cite{foucaud2022}, thus using the diameter alone as the
parameter is not fruitful. 

As a third result, we show that \SoPP becomes FPT when parameterized by both treewidth and diameter. 
To obtain this result, we use a dynamic programming scheme
analogous to that of \Cref{thm:IPP:DP-xp-tw}, but manage to reduce the number of states by storing more succinct information about the distances of the vertices to the bags of the decomposition. 

\begin{restatable}{theorem}{restateIPPdpdiamtw}
\label{thm:IPP:DP-fpt-diam-tw}
    \SoPP admits an algorithm running in time $\diam^{O(\tw^2)}\cdot n^{O(1)}$ where $\diam$ is the diameter of $G$, $\tw$ its treewidth, and $n$ its number of vertices.
\end{restatable}

We note that parametrization by both diameter and treewidth has been 
explored earlier in the context of other
metric-based problems~\cite{husfeldt:LIPIcs.IPEC.2016.16,DBLP:conf/icalp/FoucaudGK0IST24}. 

Note also that the dependency on treewidth is unusually high, as most natural problems that are \FPT\ for treewidth
admit algorithms running in time $2^{O(\tw)}\cdot n^{O(1)}$ or
$2^{O(\tw \log (\tw))}\cdot n^{O(1)}$.
We however show that an improved algorithm achieving these types of running time is highly unlikely.

\begin{restatable}{theorem}{restateIPPdiamtwhard}
\label{thm:IPP:eth-diam-tw-hard}
    Unless the \emph{\textsf{Randomized-ETH}} fails, \SoPP does not admit an algorithm running in time $\diam^{o(\pw^2/(\log^3(\pw)))} \cdot n^{O(1)}$.
\end{restatable}

We remark that this type of lower bounds, i.e., forbidding running times
roughly of the form $2^{o(p^2)}$ for some parameter $p$, matched by
an algorithm of this running time, are relatively rare in the literature. 
We refer here to the only other such results known to
us~\cite{DBLP:conf/mfcs/Pilipczuk11,DBLP:journals/iandc/SauS21,DBLP:journals/toct/AgrawalLSZ19,DBLP:conf/isaac/ChakrabortyFMT24,DBLP:conf/colt/ChalopinCIR24,DBLP:journals/corr/abs-2405-01344}
which hold for the parameters pathwidth, vertex cover number, or
solution size.

\myparagraph{Related works.} 
\SoPP (under this name or the one of \textsc{Shortest Path Partition}) was {introduced as a natural variation of the related \textsc{Isometric Path Cover}, which is motivated by applications in the cops and robber game~\cite{FF01}}. \SoPP was studied
from the structural point of view for specific graph families~\cite{DBLP:journals/dm/FitzpatrickNHC01,pmanuelisometric,penev2025isometricpathpartitionnew} and shown
to be \NP-complete in~\cite{pmanuelisometric}. This holds even for
bipartite graphs of diameter~4~\cite{PPP}, chordal graphs of
diameter~2~\cite{foucaud2022} and split
graphs~\cite{chakraborty2024covering}. \SoPP is known to be
polynomial-time solvable on
trees~\cite{goodman1974hamiltonian,boesch1974covering,kundu1976linear,franzblau_raychaudhuri_2002},
cographs~\cite{chakraborty2024covering}, and chain
graphs~\cite{chakraborty2024covering}. It can also be solved
in polynomial time for any fixed number of solution paths by \XP\
algorithms, using two different methods:
see~\cite{dumas2024graphs} and~\cite{PPP},
respectively. \SoPP is also shown to be \FPT\ when parameterized by the neighborhood diversity of the input graph, and also when parameterized by the dual parameter $n-k$~\cite{PPP}. The variant of \SoPP for DAGs is \W[1]-hard for solution size~$k$~\cite{PPP}.

The related problem \textsc{Isometric Path Cover}, where the objective is to cover the vertex set of the input graph with (not necessarily disjoint) isometric paths, has been studied recently~\cite{foucaud2022,dumas2024graphs,DBLP:conf/mfcs/ChakrabortyCFV23} and is relevant in the context of machine learning~\cite{TG21}.

Another related problem is \textsc{Disjoint Shortest Paths}, where we are
given a set of terminal pairs and we need to find disjoint isometric
paths connecting the pairs, is also studied:
see~\cite{lochet2021polynomial,DBLP:journals/siamdm/BentertNRZ23} and
references therein. A global minimization variant called
\textsc{Shortest Disjoint Paths} (where only the sum of lengths of the
path needs to be minimized) is also
studied~\cite{DBLP:journals/siamcomp/BjorklundH19,DBLP:conf/soda/MariMS24}. These
two problems are variants of the celebrated \textsc{Disjoint Paths}
problem~\cite{DBLP:journals/jct/KawarabayashiKR12}, which is central
in the theory of graph minor
testing~\cite{DBLP:journals/jct/RobertsonS95b}.

When the paths are not required to be isometric, we have the general
\textsc{Path Partition}
problem~\cite{corneil2013ldfs,goodman1974hamiltonian} (also known
under the names of \textsc{Path Cover} and \textsc{Hamiltonian
Completion}), that generalizes \textsc{Hamiltonian Path}. This
problem is also important from a structural point of view,
see~\cite{berge1983path,hartman1988variations,PP-L21bis}. Sometimes
the version where the paths need to be induced (or chordless) is also
studied, see~\cite{DBLP:journals/tcs/PanC07,PPP}.

Practical applications of path partition problems
are numerous, for example, automatic
translation~\cite{DBLP:journals/ipl/LinCL06}, network
routing~\cite{DBLP:journals/winet/SrinivasM05}, program
testing~\cite{ntafos1979path} or parallel
programming~\cite{pinter1987mapping}, to name a few.
We refer to the
surveys~\cite{DBLP:journals/dam/AndreattaM95,manuel2018revisiting} for
more references on partitioning (and covering) problems with paths.

\myparagraph{Organization of the paper.}
We start with some preliminaries in \Cref{sec:prelim}. We
present our dynamic programming schemes, proving
\Cref{thm:IPP:DP-xp-tw} and \Cref{thm:IPP:DP-fpt-diam-tw}, in
\Cref{sec:DP}. 
These algorithms are completed by a discussion in Section~\ref{sec:kernel} showing that \SoPP{} does not admit a polynomial kernel when parameterized by $\diam+\pw$, unless
$\NP \subseteq \coNP/\poly$.
We prove \W[1]-hardness
for pathwidth, \Cref{thm:IPP:w1-tw-hard}, in
\Cref{sec:IPP:tw-w1-hardness}. We prove the \textsf{randomized
\ETH}-based lower bound, \Cref{thm:IPP:eth-diam-tw-hard}, in
\Cref{sec:IPP:eth-diam-tw}. 
Finally, we conclude in \Cref{sec:conclusion}.

\section{Preliminaries}
\label{sec:prelim}

A graph has vertex set $V(G)$ and edge set $E(G)$. We denote edge with endpoints $u, v$ as $(u, v)$. 
The \emph{length} of a path $P$ is the number of edges in $P$.  An \emph{isometric path} (or \emph{IP} for short) is a shortest path between its endpoints.
An \emph{IP-partition} of a graph $G$ is a partition of the vertex set into isometric paths. 
The \emph{size} of an IP-partition $\mathcal{P}$ of a graph $G$ is the cardinality of $\calP$. 
In the rest of the paper, we shall denote an isometric path $P$ by the natural ordering $(x_1, x_2, \dots, x_{\ell})$ of its vertices obtained by traversing $P$ from one endpoint to the other.
In addition, and for simplicity if no ambiguity arises, we may either refer to $P$ as a graph, or as a set of vertices, or as a set of edges.

For a graph $G$ and a set $X\subseteq E(G)$, the graph $G+X$ (resp.\ $G-X$) is the
graph obtained by adding (resp.\ removing) the edges in $X$ to (resp.\ from) $G$. 
For an induced subgraph $H$ of a graph $G$, and a set $X\subseteq V(G)$, $H+X$ (resp.\ $G-X$) is the 
subgraph of $G$ obtained by adding (resp.\ removing) the vertices in $X$ to (resp.\ from) $H$. 
When performing these operation, we ignore multiplicity, i.e., we keep a simple graph.
For a graph $G$ and a set $X\subseteq V(G)$, the graph $G[X]$ is the subgraph of $G$ induced by the vertices in $X$.

The notation of \emph{identifying} two vertices $u$ and $v$ in a graph $G$ is defined as follows:
The operation construct a simple graph $H$ by deleting vertices $u$ and $v$
from $G$, adding a new vertex $w$ and making it adjacent to
every vertex in $G$ that were adjacent to $u$ or $v$.

For a positive integer $q$, we denote {the} set $\{1, 2, \dots, q\}$ by $[q]$.

For sake of completeness we recall the well-known definitions of \emph{tree-decompositions}, 
\emph{treewidth} and \emph{nice tree-decompositions} below. 

\begin{definition}[\cite{cygan2015parameterized}]
A \emph{tree-decomposition} of a graph $G$ is a rooted tree $T$ 
where each node $v$ is associated to a subset $X_v$ of $V(G)$ called \emph{bag}, such that 
\begin{itemize}
\item The set of nodes of $T$ containing a given vertex of $G$ forms a nonempty connected subtree of $T$; and

\item Any two adjacent vertices of $G$ appear in a common node of $T$.

\end{itemize}
The \emph{width} of $T$ is the maximum cardinality of a bag minus one. The \emph{treewidth} of $G$ is the 
minimum integer $k$ such that $G$ has a tree-decomposition of width $k$.
\end{definition}

A \emph{path-decomposition} of a graph is a tree-decomposition $T$ where $T$ is a path. The \emph{pathwidth} of a graph $G$, denoted as $\pw(G)$, is the minimum integer $k$ such that $G$ has a path-decomposition of width $k$.

\begin{definition}[\cite{cygan2015parameterized}]
A \emph{nice tree-decomposition} of a graph $G$ is a rooted tree-decomposition such that each internal 
node has one or two children, with the following properties.
\begin{itemize}
\item Each node of $T$ belongs to one of the following types: \emph{introduce}, \emph{forget}, \emph{join} or \emph{leaf}.

\item A \emph{join node} $v$ has two children $v_1$ and $v_2$ such that $X_v = X_{v_1} = X_{v_2}$.

\item An \emph{introduce node} $v$ has one child $v_1$ such that $X_v \setminus \{x\} = X_{v_1}$, where $x \in X_v$.

\item A \emph{forget node} $v$ has one child $v_{1}$ such that $X_v  = X_{v_1} \setminus \{x\}$, where $x \in X_{v_1}$.

\item A \emph{leaf node} $v$ is a leaf of $T$ with $X_v=\emptyset$.

\item The tree $T$ is rooted at a node $r$ called root node with $X_r = \emptyset$.
\end{itemize}
\end{definition}

For a node $t$ in a nice tree-decomposition $T$ of a graph $G$, we let $G_t$ denote the subgraph of $G$ induced by the vertices in the union of bags of the nodes that belong to the subtree of $T$ rooted at $t$.

We state properties that will play a crucial role in simplifying the analysis of IP-partitions.
The first property is the following, which we will often use implicitly.
Let us recall that the length of a path denotes its number of edges, not vertices.

\begin{lemma}[Leaf lemma]\label{lem:leaf-lemma}
Let $G$ be a graph, and $D$ denote the set of vertices of $G$ of degree $1$. Let $v\in V(G)$ such that $N(v)\cap D=\{u\}$. Then, $G$ has an IP-partition with minimum cardinality containing a path with $u$ as an endpoint and of length at least~1. 
\end{lemma}

\begin{proof}
Let $S$ be any IP-partition of $G$ with minimum cardinality that does not satisfy the lemma. Then $S$ contains the one-vertex path $(u)$. Let $Q$ be the path that contains $v$. Observe that $v$ cannot be an endpoint of $Q$, as otherwise we merge $(u)$ and $Q$ into one path which is still an isometric path. Otherwise, let $x$ be a vertex adjacent to $v$ on $Q$ and let $e_1$ and $e_2$ be the endpoints of $Q$, such that $x$ lies on the subpath of $Q$ between $e_1$ and $v$. We replace $Q$ and $\{u\}$ by $(u,v,\dots, e_2)$ and $(e_1,\dots, x)$. Since $Q$ is an isometric 
path, then these two paths are both isometric paths. Hence, we get another IP-partition $S'$ of $G$ with $|S|=|S'|$ that satisfies the property of the lemma. 
\end{proof}

In the remaining of the article, we denote by \emph{cherry} an induced path on three vertices with endpoints of degree~1.
We show that cherries may be assumed to be part of any optimal IP-partition. 

\begin{lemma}[Cherry lemma]\label{lem:cherry-lemma}
	Let $G$ be a graph, and $D$ denote the set of vertices of $G$ of degree~$1$. Let $v\in V(G)$ be a vertex such that $N(v)\cap D=\{u_1,u_2\}$. Then  $G$ has an IP-partition with minimum cardinality containing the path $(u_1,v,u_2)$. 
\end{lemma}

\begin{proof}
	Let $S$ be any IP-partition of $G$ with minimum cardinality that does not satisfy the lemma and without loss of generality, let $S$ contains $(u_1)$. Let $Q$ be the path that contains $v$. Observe that $v$ cannot be an endpoint of $Q$, as otherwise we merge $(u_1)$ and $Q$ into one path which is still an isometric path of $S$. Otherwise, let $e_1$ and $e_2$ be the endpoints of $Q$. 
	
	Suppose $u_2\in \{e_1,e_2\}$. Without loss of generality, let $u_2=e_2$. Let $x\in V(Q)$ be the vertex which is distinct from $u_2$ and is adjacent to $v$. Then we replace $Q$ and $(u_1)$ with $(e_1,\ldots,x)$ and $(u_1,v,u_2)$ to get the desired IP-partition that has the same cardinality as $S$. 
	
	Otherwise, $(u_2)\in S$. Let $\{x_1,x_2\}=N(v)\cap V(Q)$ such that for $i\in \{1,2\}$, $x_i$ lies between $e_i$ and $v$ in $Q$. Now we replace $Q, (u_1), (u_2)$ with $(e_1,\ldots,x_1), (e_2,\ldots,x_2), (u_1,v,u_2)$ to get the desired IP-partition that has the same cardinality as $S$.  
\end{proof}

We derive the following as a corollary of the Cherry lemma.

\begin{lemma}[Twin-cherries lemma]\label{lem:twin-cherry-lemma}
	Suppose that $G$ contains a pair of cherries with an isometric path connecting their middle vertex, and that every other vertex in $G$ is only connected to this subgraph via the middle vertices of the cherries. Then $G$ has an IP-partition with minimum cardinality containing the two cherries as well as the internal part of the path connecting their middle vertex.
\end{lemma}

Let us highlight an important behavior that the cherries achieve, and that will help simplify the proofs in our hardness reductions in Sections~\ref{sec:IPP:tw-w1-hardness} and \ref{sec:IPP:eth-diam-tw}.
As stated in Lemma~\ref{lem:cherry-lemma}, any minimum IP-partition $\calP$ of a graph $G$ can be assumed to contain any cherry.
Thus, cherries can be added to a graph while assuming that no path in a minimum IP-partition other than cherries use these newly added vertices.
Consequently, by adding a cherry to the graph, and making its middle vertex adjacent to a set $A$ of vertices, we are able to reduce the distance between elements of $A$ to at most $2$, without changing the ``structure'' of the IP-partition, in the sense that cherries can be ignored from the set of vertices that are reachable by other isometric paths.
Moreover, this metric reduction can be adapted to arbitrary distances by using twin-cherries as described in Lemma~\ref{lem:twin-cherry-lemma}.
Examples of such cherries and twin-cherries are depicted in Figures~\ref{fig:tw-w1:semi-grid}--\ref{fig:tw-w1:valves}.

\renewcommand{\C}{\mathcal{C}}
\newcommand{\I}{\mathcal{I}}
\renewcommand{\P}{\mathsf{P}}
\newcommand{\Q}{\mathcal{Q}}
\newcommand{\QQ}{\mathbf{Q}}
\renewcommand{\R}{\mathcal{R}}

\newcommand{\dd}{\mathbf{d}}

\section{Dynamic programming schemes}
\label{sec:DP}

In this section, we give a dynamic programming algorithm solving \SoPP{} in \XP{} time parameterized by the width of a given tree-decomposition.
This algorithm can be roughly described as storing, for each bag of a decomposition, the possible intersections of isometric paths partitions with the bags, an indication of how the obtained pieces are connected outside the bags, together with the location of their endpoints, which are used to ensure that the paths are isometric.
Let us point that, although this approach appears to follow standard techniques, it requires special attention when handling these indications.

Let us next introduce some terminology that will be useful in the following.
Let $P$ be a path of $G$ and $X\subseteq V(G)$ be a subset of vertices.
The \emph{trace of $P$ on $X$} is the family of paths induced by the connected components of $P[X]$.
If $\mathcal{P}$ is a family of (vertex) disjoint paths, then its \emph{trace on $X$} is the union, among all $P\in \mathcal{P}$, of the traces of $P$ on~$X$. 
Note that the trace of a family of disjoint paths also defines a family of disjoint paths.
Moreover, if $\mathcal{P}$ is a path partition of $G$, then its trace on $X$ is a path partition of $X$.

\subsection{Partial Solutions}

We introduce a notion of partial solution which will be related to the value of a state in our algorithm. 
We note that while we may have given a short and more permissive definition of a partial solution, we have chosen on purpose to provide one that is very constrained in order to simplify most of the upcoming arguments.
For a better intuition of the definition, we refer the reader to Figure~\ref{fig:partial-solution-and-signature}.

Let $T$ be a tree-decomposition of $G$ and $t$ be a node of $T$.
Let $w$ be the width of $T$.
In the definition below, by $\sigma\cup \tau$ for two functions $\sigma,\tau$ we mean the set of vertices mapped by $\sigma,\tau$. 
We call \emph{partial solution} at node $t$ a tuple $\P=(\Q,\alpha,\top,\sigma,\tau)$ where:
\begin{itemize}
    \item $\Q=\{Q_1,\dots, Q_k\}$ is an IP-partition of $G_t$;
    \item $\alpha: \Q \to \{0,1,\dots, w+1\}$ is a $(w+2)$-coloring of these paths; 
    \item $\sigma,\tau : \{1,\dots,w+1\} \to V(G)$ are two functions called \emph{terminal functions};
    \item $\top: \{1,\dots, w+1\}\to \{\{u,v\} : u,v\in X_t \cup \sigma \cup \tau\}$ is a function called \emph{linking function}; 
\end{itemize}
and such that, for each positive color $1\leq i\leq w+1$: 
\begin{itemize}
    \item the paths of color $i$ intersect $X_t$, and those of color $0$ do not;
    \item the function $\top$ maps colors to non-adjacent pairs of vertices that are either endpoints of paths of color $i$ among $\Q$, or terminal vertices in $\sigma(i)$ and $\tau(i)$;
    \item the graph $A(\P,i)$ obtained by union of the paths of color $i$ together with the edges defined by~$\top(i)$ is a $\sigma(i)$--$\tau(i)$ path; 
    moreover, it is required that such a $\sigma(i)$--$\tau(i)$ path is \emph{coherent} with $G$ in the following sense: there exists an isometric $\sigma(i)$--$\tau(i)$ path in $G$ which coincides with $A(\P,i)$ on edges different from those of $\top(i)$, and whose maximal portions not in $G_t$ precisely connect pairs of $\top(i)$. 
\end{itemize}
In the following, we call \emph{assembled path of color $i$} the graph $A(\P,i)$, and call \emph{top link of color $i$} a pair in $\top(i)$. 
We will slightly abuse notation and say that a vertex \emph{$x$ has color $i$}, denoted $\alpha(x)=i$, if $x$ belongs to a path $Q$ of color $i$.
Finally, the \emph{type} of an edge in $A(\P,i)$ is either \emph{regular} if it belongs to $G$, or of \emph{top link} type otherwise.

\begin{figure}[t]
\centering
\includegraphics[page=1,width=0.55\textwidth]{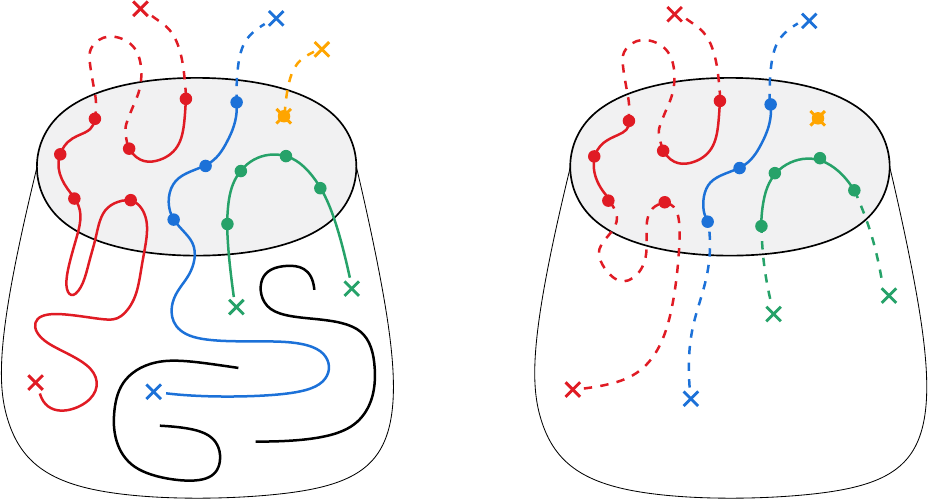}
\caption{The illustration of a partial solution (left) and a compatible signature (right). Paths in black are those of color 0, and other colors represent positive colors $1\leq i\leq w+1$. Dotted lines represent top and bottom links, and crosses represent terminal vertices.}
\label{fig:partial-solution-and-signature}
\end{figure}

Note that if $P_1,\dots, P_k$ is an IP-partition of $G$, then its trace $\Q$ on $V(G_t)$ yields a partial solution together with the following definitions of $\alpha,\top,\sigma$ and $\tau$. 
For $\alpha$, we map each $Q\in \Q$ not intersecting $X_t$ to $0$, and other $Q\in \Q$ to color $i\geq 1$ if $Q$ is a subpath of~$P_i$.
Note that the codomain of $\alpha$ is indeed restricted to $(w+2)$ colors as $|X_t|\leq w+1$ and every path of positive color intersects $X_t$.
Then for each color $1\leq i\leq w+1$ such that $\alpha^{-1}(i)\neq\emptyset$ we do the following.
We map $\sigma(i)$ and $\tau(i)$ to the endpoints of $P_i$ ensuring that $\sigma(i)\neq\tau(i)$ if $P_i$ has at least two nodes.
As for $\top$, we consider a natural ordering of the vertices of~$P_i$, and define $\top(i)$ as the set of consecutive endpoints of distinct paths of color $i$ in $\Q$ with respect to this ordering.
Intuitively, $\alpha$ indicates from which path $P_i$ the pieces of paths obtained by intersection with $G_t$ come from, while $\top$ indicates which endpoints of such pieces are linked through $G-G_t$ in~$P_i$; as of $\sigma,\tau$ they will prove useful in later representations of partial solutions and their compatibility.

The \emph{order} of a partial solution $\P$ is defined as the number of paths of color $0$ plus the number of non-empty classes of colors $1\leq i\leq w+1$.
Intuitively, this value aims to correspond to the number $k$ of isometric paths $P_1,\dots,P_k$ from which this partial solution originates in the full graph $G$.

Note that not all partial solutions correspond to the intersection of an actual IP-partition $P_1,\dots,P_k$ of $G$ with $G_t$.
This is however not an issue as we will ensure (thanks to notions of compatibility between states and nodes of the decomposition) that these partial solutions are not used by their parent nodes at some stage of the dynamic programming. 

\subsection{Signatures of Partial Solutions}\label{sec:dp-signatures}

In our dynamic programming, we will not compute all partial solutions at a given node $t$, but possible compact representations of these solutions based on their interaction with $X_t$, that we will call signatures. 
This is because the number of partial solutions, due to their intersection with $G_t-X_t$, is not bounded by $n^{f(w)}$ for any function, while we need such a bound to achieve \XP{} time.
Thus, a \emph{signature} will be a tuple $\S=(\R,\beta,\top,\bot,\sigma,\tau)$ where:
\begin{itemize}
    \item $\R=\{R_1,\dots,R_k\}$ is an IP-partition of $G[X_t]$;
    \item $\beta: \R \to \{1,\dots, w+1\}$ is a $(w+1)$-coloring of these paths; 
    \item $\sigma,\tau : \{1,\dots,w+1\} \to V(G)$ are \emph{terminal functions};
    \item $\top,\bot: \{1,\dots, w+1\}\to \{\{u,v\} : u,v\in X_t \cup \sigma\cup \tau\}$ are \emph{linking functions};
\end{itemize}
and such that, for every color $1\leq i\leq w+1$: 
\begin{itemize}
    \item the functions $\top$ and $\bot$ map colors to non-adjacent pairs of vertices that are either endpoints of paths of color $i$ among $\R$, or terminal vertices in $\sigma(i)$ and $\tau(i)$;
    \item the graph $A(\S,i)$ obtained by union of the paths of color $i$ together with the edges defined by~$\top(i)$ and $\bot(i)$ is a $\sigma(i)$--$\tau(i)$ path; 
    moreover, it is required that such a $\sigma(i)$--$\tau(i)$ path is \emph{coherent} with $G$ in the following sense: 
    there exists an isometric $\sigma(i)$--$\tau(i)$ path in $G$ which coincides with $A(\S,i)$ on edges different from those of $\top(i)$ and $\bot(i)$, and whose maximal portions not in $X_t$ either lie in $G-G_t$ in which case they connect pairs of $\top(i)$, or lie in $G_t-X_t$ in which case they connect pairs of $\bot(i)$.
\end{itemize}
Intuitively, $\R$ represents the trace of the family $\Q$ of a partial solution $(\Q,\alpha,\top,\sigma,\tau)$ on the bag~$X_t$, together with some additional constraints we describe next; see Figure~\ref{fig:partial-solution-and-signature}.
The function $\beta$ plays the same role as $\alpha$ in a partial solution except that color $0$ is not allowed anymore.
Functions $\top,\sigma,\tau$ are the same functions.
Pairs in $\bot$ play the dual role of $\top$, that is, of indicating which connections are made through $G_t-X_t$.
Similarly as for partial solutions, we call \emph{assembled path of color $i$} the path $A(\S,i)$.
We call \emph{top link of color $i$} a pair in $\top(i)$, and \emph{bottom link of color $i$} a pair in $\bot(i)$. 
The \emph{type} of an edge in $A(\S,i)$ may now be \emph{regular} if it belongs to $G$, or of \emph{top link} type if it belongs to $\top(i)$, and of \emph{bottom link} type if it belongs to $\bot(i)$.

Note that testing whether an arbitrary tuple is a signature can be conducted in $n^{O(1)}$ time, with the coherence of $A(\S,i)$ being tested by appropriately confronting the distances in $G$ from $\sigma(i)$ and $\tau(i)$ to the rest of the path, and making sure that the distances between pairs in $\top(i)$ and $\bot(i)$ are the same in $G-X_t$ or $G_t-X_t$, respectively.
Note that all the distances in the graph can be precomputed in $n^{O(1)}$-time.

Finally, we say that a signature $\S$ and a partial solution $\P$ are \emph{compatible} if, for every positive color $1\leq i\leq w+1$, the assembled paths $A(\S,i)$ and $A(\P,i)$ only differ by bottom link of $A(\S,i)$ being maximal portion of $A(\P,i)$ completely included in $G_t-X_t$. 
Note that to any partial solution $\P$ corresponds a compatible signature $\S$ that is obtained by intersecting it with $X_t$ in the natural way, and coding the parts in $G_t-X_t$ by bottom links: the coherence of the assembled paths of $\S$ naturally comes from the coherence of the assembled paths of~$\P$.
In particular $\top,\sigma$ and $\tau$ are identical for $\S$ and $\P$. 
This concludes the definition of the objects we will be manipulating in the remaining of the section.

\subsection{States of the Dynamic Programming Algorithm}

A state of our dynamic programming algorithm will be a signature 
\[
	\S=(\R,\beta,\top,\bot,\sigma,\tau)
\]
together with a node $t$ of a nice tree-decomposition $T$.
We define its value $\dd[\S,t]$ as the minimum order of a partial solution which is compatible with $\S$, and $+\infty$ otherwise.
Note that as we are aiming at \FPT{} and \XP{} running times, we can indeed assume that a state is a signature since verifying whether an arbitrary tuple is a signature can be conducted in $n^{O(1)}$ time, and we return $+\infty$ for the state otherwise.

The goal of the section is to prove that the values of the states can be computed in a bottom-up fashion, 
and that the size of an optimal solution to \SoPP can be obtained at the root of the tree-decomposition.

In the remainder of the section, we will describe modifications to perform on signatures and partial solutions, to argue that the value of a state can be computed from the value of a child state.
For better readability, we will describe these modifications directly on their associated assembled path, as we did for the definition of compatibility between a partial solution and a signature.
Formally, however, it should be understood that these modifications are to be performed on the colored path partition, the top and bottom links functions, and the terminal functions that describe these assembled paths.

\subsubsection{Leaf Node}

Let $t$ be a leaf node.
Then $X_t=\emptyset$ and the only signature consists of an empty set $\R$ and empty functions $\beta,\top,\bot,\sigma,\tau$ since they are all related to the nature of a path partition of $X_t$ which is trivially empty.
The optimal value of this state is 0 which is optimal since $G_t$ is the empty graph.
This is formalized in the next lemma.

\begin{lemma}\label{lem:dp-leaf-node}
    If $t$ is a leaf node then $\dd[\S,t]=0$ for $\S=(\emptyset,\emptyset,\emptyset,\emptyset,\emptyset,\emptyset)$, and $\dd[\S',t]=+\infty$ for every other tuple $\S'$.
\end{lemma}

\subsubsection{Introduce Vertex Node} 

Let $t$ be an introduce vertex node with child $t'$; then $X_t=X_{t'}\cup \{x\}$ for some vertex $x$ of $G-G_{t'}$.
Let $\S=(\R,\beta,\top,\bot,\sigma,\tau)$ be a signature for $t$, and let $i$ be the color of $x$.
Note that $x$ is not incident to any bottom link $xy$ as otherwise $x,y$ are separated by $X_t$ which is excluded in the definition of a signature, where it is required that there exists a $\sigma(i)$--$\tau(i)$ isometric path containing $x,y$ with the $x$--$y$ portion lying in $G_t-X_t$.
Let $\S'=(\R',\beta',\top',\bot',\sigma',\tau')$ be a signature for~$t'$.
We say that $\S'$ is \emph{introduce-compatible} with $\S$ if $A(\S,j)=A(\S',j)$ for every $j\neq i$, and if, additionally:
\begin{enumerate}[label=\textbf{I\arabic*}, ref=I\arabic*]
    \item\label{dp:intro-comp-solo} $x$ is the only vertex of its color in $\S$ and $x$ does not belong to $\S'$; or
    \item\label{dp:intro-comp-not-solo} $x$ is not the only vertex of its color in $\S$, and either: 
    \begin{enumerate}[label=\textbf{I2\alph*}, ref=I2\alph*]
        \item\label{dp:intro-comp-endpoint} $x$ is an endpoint of $A(\S,i)$, in which case $A(\S',i)$ contains $x$ as an endpoint as well (recall that endpoints of the colored paths, which are terminal vertices, may lie out of $X_t$) and $A(\S',i)$ may only differ on $A(\S,i)$ by the edge $xy$ being a top link in $\S'$; or 
        
        \item\label{dp:intro-comp-middle} $x$ is not an endpoint of $A(\S,i)$, in which case $A(\S',i)$ does not contain $x$, $x$ has two neighbors $y,z$ in $A(\S,i)$, and $A(\S',i)$ is equal to $A(\S,i)$ after replacing $yxz$ by a top link $yz$; this later situation is depicted in Figure~\ref{fig:intro-comp-middle}.
    \end{enumerate}
\end{enumerate}
In the following, we note $I(\S)$ the set of signatures that are \emph{introduce-compatible} with $\S$, and prove the following recurrence.

\begin{lemma}\label{lem:dp-introduce-node}
    If $t$ is an introduce vertex node with child $t'$ such that $X_{t'}=X_t\setminus \{x\}$ and $\S=(\R,\beta,\top,\bot,\sigma,\tau)$ is a signature of $t$, then 
    \[
        \dd[\S,t]= \min_{\S'\in I(\S)}
        \begin{cases}
        \dd[\S',t']+1 & \text{if $x$ is the only vertex of its color in $\S$,}\\
        \dd[\S',t'] & \text{otherwise.}\\
        \end{cases}       
    \]
\end{lemma}

\begin{proof}
    We start by proving the $\leq$ inequality.
    Let $\S':=(\R',\beta',\top',\bot',\sigma',\tau')$ be a signature of $t'$ that is introduce-compatible with $\S$, and $i$ be the color of $x$ in $\S$.
    If $\dd[\S',t']=+\infty$, then the inequality trivially holds.
    Otherwise, let $\P':=(\Q',\alpha',\top',\sigma',\tau')$ be a partial solution at node $t'$, compatible with signature $\S'$, and of minimum order $\dd[\S',t']$.
     
    We first assume that $x$ is the only vertex of its color in $\S$.
    Then $\S'$ satisfies Condition~\ref{dp:intro-comp-solo}, that is, $x$ does not belong to $\S'$.
    Moreover by the compatibility of $\S$ and $\S'$, no other element has color $i$ in $\S'$.
    We create a partial solution $\P$ of order $\dd[\S',t']+1$ for $t$ from $\P'$ by setting $A(\P,j):=A(\P',j)$ for every color $j$ different from $i$, and setting $A(\P,i):=A(\S,i)$. 
    Clearly this defines an IP-partition of $G_t$ and the coherence of $A(\P,i)$ follows from that of $A(\S,i)$.
    Hence we indeed obtain a partial solution. 
    Its compatibility with $\S$ also follows from $A(\P,i)=A(\S,i)$ as $\S'$ and $\P'$ are compatible and only differ on color $i$ with $\S$ and $\P$.
    Hence we conclude that $\dd[\S,t]\leq \dd[\S',t']+1$ in that case.

    \begin{figure}[t]
    \centering
    \includegraphics[page=2,width=0.55\textwidth]{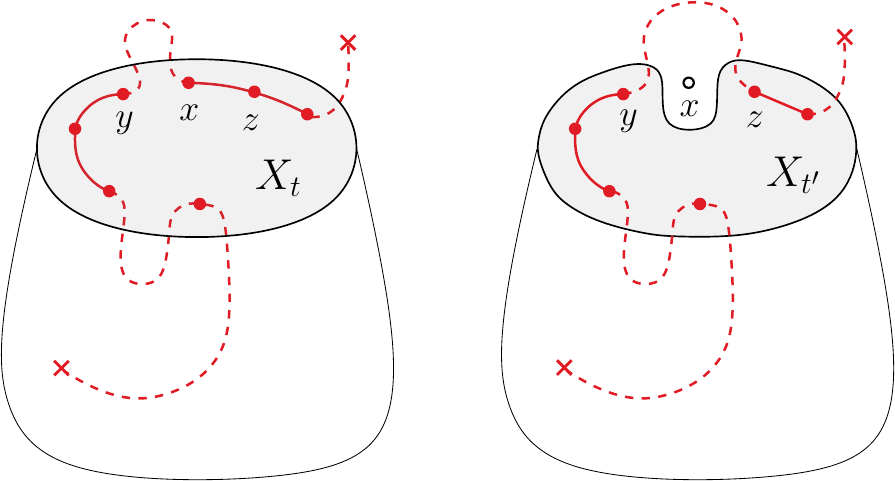}
    \caption{An illustration of Condition \ref{dp:intro-comp-middle}. Parts of the signature $\S$ are depicted on the left, and the corresponding compatible parts of $\S'$ are depicted on the right.}
    \label{fig:intro-comp-middle}
    \end{figure}

    Let us now assume that $x$ is not the only vertex of its color in $\S$.
    Then $\S'$ satisfies Condition~\ref{dp:intro-comp-not-solo}.
    We create a partial solution $\P$ of same order for $t$ by modifying $\P'$ as follows:
    \begin{itemize}
        \item In Case \ref{dp:intro-comp-endpoint} we change the type of the edge $xy$ to regular it if it is regular in $\S$;
        \item In Case \ref{dp:intro-comp-middle} we remove the top link $yz$, insert $x$ between $y$ and $z$, and add the edges $xy$ and $yz$, making sure they are of the same type as in $\S$; note that this operation may amount to merge two paths of $\Q'$ along $x$.
    \end{itemize}
    The fact that $\P$ is indeed a partial solution follows from the fact that the obtained graph $A(\P,i)$ is a path by construction, and its coherence follows from that of the signature $\S$ which has the same intersection on $X_t$.
    As of the compatibility of $\S$ and $\P$, it follows from the compatibility of $\S'$ and $\P'$ together with the fact that the changes made do not concern bottom links. 
    Hence, $\dd[\S,t]\leq \dd[\S',t']$ in that case, concluding the first part of the proof.

    \medskip

    We now move to the $\geq$ inequality.
    Let $i$ be the color of $x$ in $\S$, and let $\P$ be a partial solution of $t$ compatible with $\S$.
    We first assume that $x$ is the only vertex of its color in $\S$.
    Then $x$ has no neighbor among $X_t$ in $A(\S,i)$, and since $x$ is not incident to a bottom link we derive that $\sigma(i)$ and $\tau(i)$ lie in $G-G_{t'}$, with $x$ possibly being one of such terminals, or both. 
    Note as well that $A(\S,i)=A(\P,i)$.
    A partial solution $\P'$ and signature $\S'$ for $t'$ are trivially obtained by removing $A(\P,i)$ from $\P$, and $A(\S,i)$ from $\S$.
    Their compatibility trivially follows from the fact that we removed a complete color.
    Moreover, $\S$ and $\S'$ are introduce-compatible as they satisfy Condition~\ref{dp:intro-comp-solo}.
    We conclude that $\dd[\S,t]\geq \dd[\S',t']-1\geq \min_{\S''\in I(\S)} \dd[\S'',t']-1$, proving the desired inequality in that case.

    Let us now assume that $x$ is not the only vertex of its color in $\S$.
    Then, $x$ has at least one and at most two neighbors among $X_t$ in $\P$.
    A partial solution $\P'$ for $t'$ is obtained by replacing the edges incident to $x$ by top links in $A(\P,i)$.
    Note that the obtained $\P'$ indeed defines a partial solution as $G_{t'}$ stays covered, $A(\P',i)$ is a path, and its coherence follows from that of $\S$.
    We construct a signature $\S'$ from $\S$ by performing the same modifications, whose compatibility trivially follows from the fact that $x$ is not incident to bottom links.
    Moreover, $\S$ and $\S'$ are introduce-compatible as they satisfy Conditions~\ref{dp:intro-comp-endpoint} and \ref{dp:intro-comp-middle} depending on the degree of $x$ in $A(\S,x)$.
    We derive that $\dd[\S,t]\geq \dd[\S',t']\geq \min_{\S''\in I(\S)} \dd[\S'',t']$ concluding this case and the proof.
\end{proof}

\subsubsection{Forget Vertex Node}

Let $t$ be a forget vertex node with child $t'$; then $X_t=X_{t'}\setminus \{x\}$ for some vertex $x$ in $X_t$.
Let $\S=(\R,\beta,\top,\bot,\sigma,\tau)$ be a signature for $t$.
Note that analogously as for the introduce vertex node, we may assume here that $x$ is \emph{not} incident to a top link, by definition of a signature.

Let $\S'=(\R',\beta',\top',\bot',\sigma',\tau')$ be a signature for $t'$.
We define a notion of compatibility for forget nodes that is analogous to the one for introduce vertex nodes with the top links being replaced by bottom links.
We say that $\S'$ is \emph{forget-compatible} with $\S$ if $A(\S,j)=A(\S',j)$ for every $j\neq i$, and if, additionally:
\begin{enumerate}[label=\textbf{F\arabic*}, ref=F\arabic*]
    \item\label{dp:forget-comp-solo} $x$ is the only vertex of its color in $\S'$ and $x$ does not belong to $\S$;
    \item\label{dp:forget-comp-not-solo} $x$ is not the only vertex of its color in $\S'$, and either: 
    \begin{enumerate}[label=\textbf{F2\alph*}, ref=F2\alph*]
        \item\label{dp:forget-comp-endpoint} $x$ is an endpoint of $A(\S',i)$, in which case $A(\S,i)$ contains $x$ as an endpoint as well, and $A(\S,i)$ may only differ on $A(\S',i)$ by the edge $xy$ being a bottom link in $\S$; or
        \item\label{dp:forget-comp-middle} $x$ is not an endpoint of $A(\S',i)$, in which case $A(\S,i)$ does not contain $x$, $x$ has two neighbors $y,z$ in $A(\S',i)$, and $A(\S,i)$ is equal to $A(\S',i)$ after replacing $yxz$ by a bottom link $yz$.
    \end{enumerate}
\end{enumerate}
In the following, we note $F(\S)$ the set of signatures that are \emph{forget-compatible} with $\S$, and prove the following recurrence.

\begin{lemma}\label{lem:dp-forget-node}
    If $t$ is a forget vertex node with child $t'$ and $\S$ is a signature of $t$, then 
    \[
        \dd[\S,t]=\min_{\S'\in F(\S)} \dd[\S',t'].
    \]
\end{lemma}

\begin{proof}
    We start with proving the $\leq$ inequality. 
    Let $\S'$ be a signature of $t'$ that is forget-compatible with $\S$.
    If $\dd[\S',t']=+\infty$, then the inequality trivially holds.
    Let us assume otherwise that $\dd[\S',t']\neq+\infty$, and let $\P'$ be a partial solution at node $t'$, compatible with $\S'$, and of minimal order.
    Note that $G_t=G_{t'}$.
    Furthermore as $x$ is not incident to any top link, $\P'$ defines a partial solution for $t$ as well.
    We argue that it is compatible with~$\S$.
    Since $\S$ and $\S'$ are forget-compatible, $A(\S,j)=A(\S',j)$ for every $j\neq i$ where $i$ is the color of $x$ in $\S'$.
    We may thus focus on color $i$.
    If $\S'$ satisfies Condition~\ref{dp:forget-comp-solo} then color $i$ does not exist in $\S$ and the partial solution is trivially compatible.
    Otherwise if $\S'$ satisfies Condition~\ref{dp:forget-comp-endpoint} then $A(\S,i)$ and $A(\P',i)$ may only disagree on the edge that is incident to $x$.
    However since this edge is not a top link in $A(\S',i)$, the corresponding portion in $A(\P',i)$ either lies in $X_t$ (as an edge) or in $G_{t'}-X_t$.
    In both cases it is compatible with $xy$ being a bottom link in $A(\S,i)$.
    The same argument holds for Condition~\ref{dp:forget-comp-middle} with edges $yx$ and $xz$ being either regular edges or bottom link in $A(\S',i)$, and hence the $y$--$z$ portion of $A(\P',i)$ being compatible with $yz$ being a bottom link in $A(\S,i)$.
    We conclude that $\P'$ is compatible with $\S$, hence that $\dd[\S,t]\leq \dd[\S',t']= \min_{\S''\in F(\S)} \dd[\S'',t']$, proving the desired inequality.
    
    We now prove the $\geq$ inequality.
    Let $\P$ be a partial solution at node $t$, compatible with~$\S$, and of order $k$.
    Let $i$ be the color of $x$ in $\P$, with possibly $i=0$.
    If $i=0$ we trivially get a partial solution $\P'$ and a compatible signature $\S'$ by picking an extra color $\ell$, setting $A(\P',\ell)$ to be the path of $\P$ containing $x$, and setting $A(\S',\ell)$ to be its natural representation with $\sigma(\ell)$--$x$ and $x$--$\tau(\ell)$ portions being bottom links, if any.
    Moreover, this signature $\S'$ is forget-compatible with $\S$ by Condition~\ref{dp:forget-comp-solo}.
    Note that the order of $\P'$ is equal to that of $\P$ as we removed a path of color $0$ and replaced it by one of a new color.
    Hence $\dd[\S,t]\geq \dd[\S',t]$ in that case. 
    In the other case where $i\neq 0$, $\P$ defines partial solution for $t'$ as $G_t=G_{t'}$ and every element in $X_{t'}$ has a positive color.
    Moreover, note that $x$ is not the only vertex of his color in $\P$ or $\S$.
    We create a signature $\S'$ for $t'$ as follows:
    \begin{itemize}
        \item If $x$ is an endpoint of $A(\S,i)$ then its incident edge $xy$ is a bottom link and we replace it by either a regular edge in $A(\S',i)$ if $x$ and $y$ are adjacent, or by a bottom link otherwise;
        \item If $x$ is not an endpoint of $A(\S,i)$, then it does not belong to $A(\S,i)$. For $y,z$ the closest vertices from $x$ in $A(\P,i)$ that lie in $X_t\cup \{\sigma(i),\tau(i)\}$, we remove the bottom link $yz$ from $A(\S,i)$ and add the edges $yx,xz$ if they are part of $G$, or add these pairs as bottom links otherwise, in order to obtain $A(\S',i)$.
    \end{itemize}
    Note that the obtained signature $\S'$ belongs to $F(\S)$ in each of these case as they satisfy Conditions~\ref{dp:forget-comp-endpoint} and \ref{dp:forget-comp-middle}.
    The compatibility with $\P$ follows from the fact that the modifications performed on the assembled path depends on the nature of the edges incident to $x$ in $A(\P,i)$. 
    We derive $\dd[\S,t]\geq \dd[\S',t']\geq \min_{\S''\in F(\S)} \dd[\S'',t']$, which concludes the proof.
\end{proof}

\subsubsection{Join Node}

Let $t$ be a join node with children $t_1$ and $t_2$; then $X_t=X_{t_1}=X_{t_2}$.
Let $\S=(\R,\beta,\top,\bot,\sigma,\tau)$ be a signature for $t$, and $\S_1=(\R_1,\beta_1,\top_1,\bot_1,\sigma_1,\tau_1)$ and $\S_2=(\R_2,\beta_2,\top_2,\bot_2,\sigma_2,\tau_2)$ be signatures for $t_1$ and $t_2$, respectively.

We say that $\S_1$ and $\S_2$ are pairwise \emph{join-compatible} if
$\R_1=\R_2$,
$\beta_1=\beta_2$,
$\sigma_1=\sigma_2$, 
$\tau_1=\tau_2$, and additionally,
$\bot_1\subseteq \top_2$ and
$\bot_2\subseteq \top_1$. 
Then, for every color $i$ with $1\leq i\leq w+1$, we have that $A(\S_1,i)$ and $A(\S_2,i)$ differ on the types of their bottom and top links only, whenever $\S_1$ and $\S_2$ are \emph{join-compatible}.
We say that $\S$ is \emph{join-compatible} with a pair $\{\S_1,\S_2\}$ if $\S_1$ and $\S_2$ are pairwise join-compatible, and
$\R=\R_1=\R_2$,
$\beta=\beta_1=\beta_2$,
$\sigma=\sigma_1=\sigma_2$,
$\tau=\tau_1=\tau_2$, and additionally,
$\bot=\bot_1\cup \bot_2$, and
$\top=\top_1\cap \top_2$.
Also here, for each color $i$, we derive that $A(\S,i)$, $A(\S_1,i)$ and $A(\S_2,i)$ differ on the types of their bottom and top links only, whenever $\S$ and $\{\S_1,\S_2\}$ are \emph{join-compatible}.
In the following, we denote by $J(\S)$ the family of pairs $\{\S_1,\S_2\}$ of signatures that are join-compatible with the signature $\S$, and prove the following recurrence.

\begin{lemma}\label{lem:dp-join-node}
    If $t$ is a join node with children $t_1,t_2$ and $\S=(\R,\beta,\top,\bot,\sigma,\tau)$ is a signature of~$t$, then 
    \[
        \dd[\S,t]=\min_{\{\S_1,\S_2\}\in J(\S)} 
        \left\{\dd[\S_1,t_1] + \dd[\S_2,t_2] - |\{i : \beta^{-1}(i) \neq \emptyset\}|\right\}.
    \]
\end{lemma}

\begin{proof}
    We start by proving the $\leq$ inequality.
    Let $\S_1,\S_2$ be two signatures such that $\{\S_1,\S_2\}\in J(\S)$.
    If $\dd[\S_1,t_1]=+\infty$ or $\dd[\S_2,t_2]=+\infty$ then the inequality trivially holds.
    Otherwise, let $\P_1$ and $\P_2$ be partial solutions at nodes $t_1$ and $t_2$, compatible with $\S_1$ and $\S_2$, and of minimal order $k_1$ and $k_2$, respectively.
    Consider the tuple $\P$ consisting of every path of color $0$ in $\P_1$ or $\P_2$, plus, for each positive color $i$, the \emph{merging} of $A(\P_1,i)$ and $A(\P_2,i)$ resulting in $A(\P,i)$ and defined as the union of their regular edges and mutual top link, i.e., those edges in $\top_1\cap \top_2$.
    We now argue that this operation defines a partial solution of the desired order.

    \begin{claim}
        $\P$ is a partial solution of order $k_1+k_2-|\{i : \alpha^{-1}(i)\neq \emptyset,\ i>0\}|$ for node $t$.
    \end{claim}

    \begin{proof}
        Let us focus on some color $i>0$ and first argue that $A(\P,i)$ is a $\sigma(i)$--$\tau(i)$ path.
        Recall that $A(\P_1,i)$ and $A(\P_2,i)$ are $\sigma(i)$--$\tau(i)$paths, respectively.
        Moreover, as $\bot_1\subseteq \top_2$, the parts of $A(\P_1,i)$ connecting two endpoints $x,y\in X_t\cup \sigma(i)\cup \tau(i)$ with internal vertices in $G_{t_1}-X_t$ are such that $\{x,y\}\in \top_2(i)$, and symmetrically for $A(\P_2,i)$.
        Furthermore, these parts precisely consist of their regular edges.
        Hence by connecting these pieces the operation yields a $\sigma(i)$--$\tau(i)$ path.
        We note that each vertex of $A(\P_1,i)$ and $A(\P_2,i)$ is part of $A(\P,i)$, hence that $\P$ codes an IP-partition of $G_t$.

        Let us now analyze the order $k$ of $\P$.
        The number of paths of color $0$ in $\P$ is equal to the number of paths of color $0$ in $\P_1$ and $\P_2$, which is 
        \[ 
            k_1-|\{i : \alpha_1^{-1}(i) \neq \emptyset,\ i>0\}|+k_2-|\{i : \alpha_2^{-1}(i) \neq \emptyset,\ i>0\}|.
        \]
        As of the number of paths of color $i>0$ in $\P$, it is equal to $|\{i : \alpha_1^{-1}(i) \neq \emptyset,\ i>0\}|$.
        By summing up the two, we get the desired value.
        \claimqedhere
    \end{proof}

    \begin{claim}
        $\P$ and $\S$ are compatible and $|\{i : \alpha^{-1}(i)\neq \emptyset,\ i>0\}|=|\{i : \beta^{-1}(i) \neq \emptyset\}|$.
    \end{claim}

    \begin{proof}
        We focus on an assembled path $A(\P,i)$ for some color $i>0$.
        Recall that this assembled path originates from two assembled paths $A(\P_1,i)$ and $A(\P_1,i)$ who are described by their traces $A(\S_1,i)$ and $A(\S_2,i)$ on $X_t$, respectively.
        Since $\S_1$ and $\S_2$ are compatible with $\P_1$ and $\P_2$ we derive that $A(\P,i)$ and $A(\S,i)$ satisfy the compatibility constraint on the intersection with $G_{t_1}-X_t$ and $G_{t_2}-X_t$, respectively, hence on the intersection with $G_t-X_t$.
        As of the top links in $\S$, they are precisely those in $\top_1\cap \top_2$ by join-compatibility.
        Hence $\S$ and $\P$ are compatible. 
        The second part of the statement follows from the fact that signatures are defined on positive colors and that paths of color $0$ in partial solution are required to be disjoint from $X_t$.
        \claimqedhere
    \end{proof}

    By the two claims above, we derive that $\dd[\S,t]\leq \dd[\S_1,t_1]+\dd[\S_1,t_1]-|\{i : \beta^{-1}(i) \neq \emptyset\}|$, hence to the desired inequality by minimality over $\{\S_1,\S_2\}\in J(\S)$.
    
    We now move to the $\geq$ inequality.
    Let $\P$ be a partial solution at node $t$, compatible with $\S$, and of order $k$.
    Let $\P_1$ and $\P_2$ be the tuples obtained by intersecting $\P$ on $G_{t_1}$ and $G_{t_2}$ in the natural way, where $\top_1$ is the union of the pairs in $\top$ plus the pair only separated by vertices of $G_{t_2}-X_t$ in $A(\S,i)$, and analogously for $\top_2$.

    \begin{claim}
        $\P_1$ and $\P_2$ are partial solutions at nodes $t_1$ and $t_2$.
    \end{claim}

    \begin{proof}
        Let us focus on $\P_1$ and $t_1$ for the rest of the proof as the other case is symmetric.
        Note that every vertex of $G_{t_1}$ that was covered by an assembled path of $\P$ stays covered by an assembled path by definition of $\P_1$, and those covered by paths of color $0$ are as well covered since those paths are not modified.
        As of the coherence of assembled paths $A(\P_1,i)$, $i>0$, it follows from the fact that $A(\P,i)$ is coherent, and that $A(\P_1,i)$ is equal to $A(\P,i)$ up to isometric portions being replaced by top links.
        \claimqedhere
    \end{proof}

    Analogously, let $\S_1$ and $\S_2$ be the signatures obtained by intersection of $\S$ on $G_{t_1}$ and $G_{t_2}$ in the natural way, with a link $uv$ of $\S_1$ being a bottom link if the $u$--$v$ portion of $A(\P,i)$ lies in $G_{t_1}$, a top link otherwise, and symmetrically for $\S_2$.
    Note that $\S_1$ and $\S_2$ are pairwise join-compatible by this construction, and together they are join-compatible with $\S$.
    Moreover these signatures are compatible with $\P_1$ and $\P_2$.
    We derive that $\dd[\S, t]\geq\dd[\S_1,t_1]+\dd[\S_2,t_2]-|\{i : \alpha^{-1}(i)\cap X_t \neq \emptyset\}|$ by construction, hence to the desired inequality by minimality over $\{\S_1,\S_2\}\in J(\S)$.
\end{proof}

\subsection{Proof of \Cref{thm:IPP:DP-xp-tw}}

Let us first show that the optimal size of a solution is found at the root node of the tree-decomposition.

\begin{lemma}\label{lem:dp-root-node}
    The minimum size of an IP-partition of $G$ is equal to $\dd[\S_0,t]$ where $\S_0=(\emptyset,\emptyset,\emptyset,\emptyset,\emptyset,\emptyset)$ and $t$ is the root of the decomposition.
\end{lemma}

\begin{proof}
    Since $t$ is the root, $X_t=\emptyset$ and the only signature consists of an empty set $\R$ and empty functions $\beta,\top,\bot,\sigma,\tau$, i.e., it is $\S_0$.
    Partial solutions compatible with this state are precisely those of the form $Q_1,\dots,Q_k$ with $\alpha$ mapping each $Q_i$ to color $0$, which thus define actual IP-partitions of $G_t$.
    We conclude that the minimum value of an IP-partition is given by $\dd[\S_0,t]$, proving the lemma.
\end{proof}

We are now ready to describe our dynamic programming algorithm.

The algorithm first computes (and stores) the values of each possible state in a bottom-up fashion, starting with the leaf nodes of the decomposition, and finishing at its root node $t$. 
Then it outputs the value of $\dd[\S_0,t]$ for $\S_0=(\emptyset,\emptyset,\emptyset,\emptyset,\emptyset,\emptyset)$.
The fact that each state can be correctly computed follows by induction on the type of nodes in the decomposition, relying on Lemmas~\ref{lem:dp-leaf-node}--\ref{lem:dp-join-node}. 
The fact that the solution to \SoPP{} is obtained follows by Lemma~\ref{lem:dp-root-node}.
This concludes the correctness of the algorithm.

We now turn to proving that the algorithm runs within the claimed \XP{} time, which essentially follows from the next bound on the number of states and the time to compute their values.

\begin{lemma}\label{lem:dp-number-signatures}
    The number of distinct signatures is bounded above by $n^{O(w)}$ and the value of one state can be computed within this time.
\end{lemma}

\begin{proof}
    Let us analyze the composition of a signature $\S=(\R,\beta,\top,\bot,\sigma,\tau)$.
    First note that the pair $\R,\beta$ defines a partition of $X_t$ into colored paths, and that the number of such partitions is bounded above by the number of linear orderings of the vertices of a bag (representing the concatenation of a set of paths), times the number of possible intervals in this ordering (representing each color), times the number of possible intervals within these intervals (representing the actual endpoints of the colored paths).
    Thus, the number of distinct values for $\R,\beta$ can be roughly bounded by $\tw! \cdot 2^\tw\cdot 2^\tw$ which is $2^{O(\tw \log \tw)}$.
    As for $\top,\bot$ they map each color to $(w+3)^2$ values each for a total number of $2^{O(w \log w)}$ distinct values.
    Finally, $\sigma$ and $\tau$ map each color class to at most $2$ vertices among $n$ vertices, for a total number of $n^{O(w)}$ distinct possible values.
    We conclude that the number of distinct signatures is bounded above by $2^{O(w \log w)}\cdot n^{O(w)}$ which is bounded by $n^{O(w)}$ since $w\leq n$.

    Let us now analyze the complexity of computing the value of a state $\S,t$.
    First, we need to check whether $\S$ is indeed a valid state in $n^{O(1)}$ time.
    Then, we rely on Lemmas~\ref{lem:dp-introduce-node}--\ref{lem:dp-join-node} to compute the value of $\S,t$ given the values of its children that we have already computed by induction.
    The base case corresponds to the leaf nodes, whose values of states can be initialized in $n^{O(1)}$ time per state using Lemma~\ref{lem:dp-leaf-node}.
    For the other cases, since iterating through the states of the (at most two) children takes $n^{O(w)}$ time, and checking the compatibility of each state takes time polynomial in $n$, we obtain the desired time bound.
\end{proof}

Finally, since a tree-decomposition of width $w\leq 2\tw$ can be computed in time $2^{O(\tw)}\cdot n$ \cite{DBLP:conf/focs/Korhonen21}, and can be transformed into a nice tree-decomposition with $O(n)$ nodes in $O(wn)$ time~\cite{bodlaender2013fine,kloks1994treewidth}, we derive Theorem~\ref{thm:IPP:DP-xp-tw} that we restate here.

\restatethmIPPdpxptw*

\subsection{Proof of \Cref{thm:IPP:DP-fpt-diam-tw}}

Note that in our dynamic programming algorithm, the dependence on $n$ in the number of states comes from the fact that we store the guessed endpoints of paths of an actual solution, and that these endpoints may lie anywhere in $G$.
Moreover, note that the role of these endpoints $\sigma,\tau$ is limited to check the conditions of a signature, that is, whether the distances from $\sigma,\tau$ to the traces of paths on the bag of a node, together with the side on which they live, are coherent.
Thus, we can improve the time bound of the algorithm to \FPT{} if we can perform these verifications without storing arbitrary endpoints.
In this section, we show that we are able to do so if we parameterize by the diameter of the graph in addition to the treewidth, proving Theorem \ref{thm:IPP:DP-fpt-diam-tw} that we restate here.

\restateIPPdpdiamtw*

The idea is to observe that the number of all possible distances from a vertex of $G$ to the elements of a bag $X_t=\{x_1,\dots,x_k\}$ of the decomposition is limited when the diameter is small. 
More formally, call \emph{distance profile} of $u$ to $X_t$ the vector of size $k$ where the $i^\text{th}$ coordinate denotes the distance from $u$ to $x_i$ in $G$. 
Note that there are at most $(\diam+1)^{w+1}$ such distance profiles, where $w$ is the width of the decomposition. 
Hence there are at most $(\diam+1)^{w+1}$ equivalence classes in $V(G)$ with respect to their distance to $X_t$, that we can further refine by splitting each equivalence class into two depending on whether their elements lie in the bottom part $G_t-X_t$, or in the top part $G-G_t$. 
The number of such refined classes is $\diam^{O(w)}$ and for each class we can name one vertex (say the one of smallest index) to be the \emph{representative} of the class.
Then, the function $\sigma,\tau$ need only to map to these representatives to be able to check the validity of a signature in polynomial time in $n$.
Conducting the same counting as in the proof of Lemma~\ref{lem:dp-number-signatures} we derive that the number of states in the algorithm becomes bounded by $2^{O(w \log w)}\cdot \diam^{O(w^2)}$
which is $\diam^{O(w^2)}$.
Computing these equivalence classes and their representative takes $\diam^{O(w^2)}\cdot n^{O(1)}$ time.
Testing the validity of a signature, as well as testing its compatibility with the signature of a neighboring node in the decomposition, stays the same and requires polynomial time in $n$.
Summing up, the value of each state can be computed in $\diam^{O(w^2)}\cdot n^{O(1)}$ time.
Now since we are able to compute the diameter and the treewidth of $G$ within this time \cite{bentert2023parameterized} we derive a $\diam^{O(\tw^2)}\cdot n^{O(1)}$ time algorithm for \SoPP{}, as desired.

\section{Hardness with Respect to Pathwidth}
\label{sec:IPP:tw-w1-hardness}

In this section, we show that the \SoPP
problem is \W[1]-hard when parameterized by the
pathwidth (and hence treewidth) of the input graph,
thereby proving Theorem~\ref{thm:IPP:w1-tw-hard}.
We present a parameterized reduction from
\textsc{Multicolored Clique}, which is \W[1]-hard when parameterized
by the solution size; see, e.g., \cite[Chapter 13]{cygan2015parameterized}.

\medskip

\defproblem{\textsc{Multicolored Clique}}{A graph $G$, an integer $k$, and a
partition $(V_1, V_2, \dots, V_k)$  of $V(G)$ such that $V_i$
is an independent set and $|V_i| = n$, i.e., $V_i = \{v^i_1, \dots
v^i_n\}$, for every $i \in [k]$.}{Does there exist a clique in $G$ containing one
vertex from $V_i$ for every $i \in [k]$?}

\medskip

In an instance of \textsc{Multicolored Clique}, the sets $V_1,\dots,V_k$ are called \emph{color classes}, and the goal can be rephrased as deciding whether there exists a multicolored clique in $G$.

\subsection{Overview of the Reduction}
\label{sec:IPP:tw-w1:overview}

The reduction takes as input an instance
$(G, k, (V_1, V_2, \dots, V_k))$ of
\textsc{Multicolored Clique} and returns an instance $(H, \text{poly}(n,k))$ of
\SoPP in polynomial time where $H$ is of polynomial size in $k$ and $n$. 
The graph constructed has pathwidth $O(k^2)$. 
For a better comprehension of the coming reduction, it is convenient here to interpret the \textsc{Multicolored Clique} problem as selecting $\binom{k}{2}$ edges in a way that this set is incident to exactly one vertex of each set $V_i$. 

The structure of $H$ is organized as follows. 
For each $1\leq i\leq k$, there is a \emph{semi-grid} $\Gamma_i$ that corresponds to the set $V_i$. 
A semi-grid has a grid-like structure with $O(k)$ rows and $O(n)$ columns. 
See Figure~\ref{fig:tw-w1:semi-grid} for an illustration. 
The columns of $\Gamma_i$ correspond to the vertices in $V_i$. 
Semi-grids are connected together via their left and right boundaries by gadgets encoding edges, as shown in Figure \ref{fig:tw-w1:edge-gadget}: for each edge $(u,v)$ with $u\in V_i$ and $v\in V_j$, there is an edge gadget that connects $\Gamma_i$ and $\Gamma_j$ on appropriate rows, whose indices depend on $i$ and $j$.

To facilitate the analysis of the reduction, a number of cherries are part of the construction. 
Recall that by Lemma~\ref{lem:cherry-lemma}, we can always assume an IP-partition of minimum cardinality to contain such cherries.
Once we made this assumption, there are only two relevant ways of partitioning an edge gadget encoding $(u,v)\in E(G)$: 

\begin{itemize}
    \item Either we decide to partition the edge gadget optimally, i.e., by selecting four isometric paths, as shown in Figure~\ref{fig:tw-w1:non-extendable-paths}. Choosing this red partition corresponds to \emph{not} selecting $(u,v)$ in the \textsc{Multicolored Clique} solution. In this case, no other vertices can be covered by such paths: we say that they are \emph{non-extendable}. 
    
    \item Or, we decide to partition the gadget with five isometric paths, as shown in Figure~\ref{fig:tw-w1:extendable-paths}. Choosing this green partition corresponds to selecting the edge $(u,v)$ in the \textsc{Multicolored Clique} solution. Although this is not an optimal way to cover the gadget, choosing the green partition has a strong benefit: paths emerging from the gadget can penetrate inside semi-grids $\Gamma_i$ and $\Gamma_j$ to cover some vertices there, which potentially reduces the number of paths needed to cover the semi-grids. 
    Penetrating inside semi-grids in that way, however, can only be done up to some column that is left uncovered, and that is referred to as the \emph{crest} in Figure~\ref{fig:tw-w1:edge-gadget}. 
    In the ideal scenario where the instance of the \textsc{Multicolored Clique} problem is positive, almost all vertices of the semi-grids are covered by the green paths emerging from a selection of edges, except for the crest columns. For each semi-grid, this column corresponds to the vertex picked in the \textsc{Multicolored Clique} solution, and can be covered using only one isometric path. 
\end{itemize}

\myparagraph{Organization.} The rest of the section is divided in three parts. 
In Section~\ref{sec:IPP:tw-w1:construction}, we give the full description of the reduction, that we break by subsections for each gadget and the value of the solution size. 
Then, we give some useful properties and observations in Section~\ref{sec:IPP:tw-w1:properties}. 
We conclude with the proof of Theorem~\ref{thm:IPP:w1-tw-hard} in Section~\ref{sec:IPP:tw-w1:proof}.

\subsection{Reduction}\label{sec:IPP:tw-w1:construction}

\subsubsection{Semi-grids for Color Classes}\label{sec:IPP:tw-w1:construction:semi-grid}

\begin{figure}[t]
    \centering
    \includegraphics[page=1,width=0.8\textwidth]{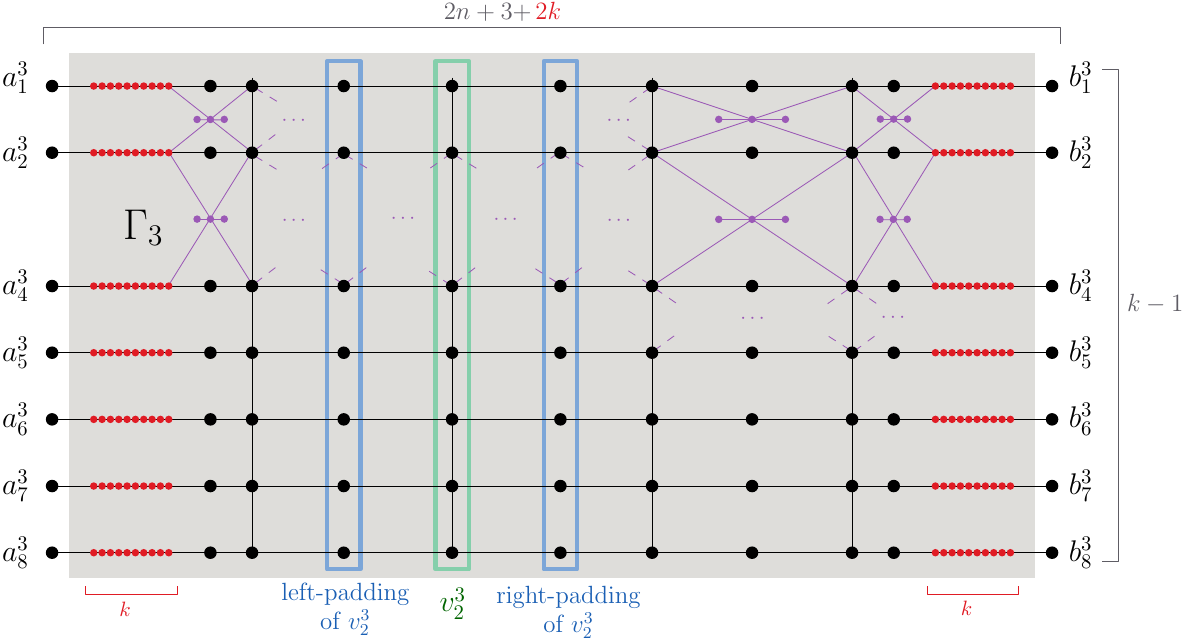}
    \caption{The final stage of the semi-grid used to represent the color class $V_i$ in the reduction, with $i=3$ and $k=8$ here.
    Each vertical path corresponds to a vertex in the set $V_i$.
    Each row will be used to connect $V_i$ to another color class $V_j$, $j \neq i$, with edge gadgets encoding all the adjacencies between $V_i$ and $V_j$; this is why no row is indexed $3$ in $\Gamma_3$, and that the number of rows is $k-1$.
    Grid cherries are depicted in purple.}
    \label{fig:tw-w1:semi-grid}
\end{figure}

Recall that $k$ is the size of the partition of $V(G)$ and $n$ is the number of 
vertices in each part. 
For every $i \in [k]$, we add a structure that we call \emph{semi-grid}, denoted by $\Gamma_i$, aimed at representing the color class $V_i$, and that we shall define now.
For convenience, we choose to present it starting with a $(k-1) \times (2n + 3)$ grid, and later specify edges to be deleted, and subdivisions to be made.
The following steps are better understood accompanied with Figure~\ref{fig:tw-w1:semi-grid} which depicts the resulting construction.

\begin{itemize}
    \item We start with a $(k-1)\times (2n + 3)$ grid $\Gamma_i$ (represented by fat black vertices in Figure~\ref{fig:tw-w1:semi-grid}) where the $k-1$ rows are purposely labeled with integers in $\{1,\dots,i-1\}$ and $\{i+1, \dots, k\}$, an indexing which will be convenient in the following.
    In other words, $\Gamma_i$ should be thought of as the $k \times (2n + 3)$ grid with rows labeled from $1$ to $k$, and on which the row $i$ has been removed, and adjacent rows made connected, while keeping the initial indexing. 
    The indexing of the columns is $0,1,\dots,2n+2$.
    Only the vertices in the leftmost and rightmost columns of the semi-grid will be adjacent to external vertices. 
    They are called the \emph{left} and \emph{right borders} of $\Gamma^i$, and are denoted by $a^i_{j}$'s and $b^i_{j}$'s, respectively, with $j$ being the index of the row.
    
    \item Each column of even index corresponds to a vertex in $V_i$.
    Formally, consider a labeling $\{v^i_1,\dots,v^i_n\}$ of the vertices of $V_i$.
    Then the $(2p)^\th$ column corresponds to the vertex $v^i_p$ for each $p \in [n]$.
    We refer to the $(2p - 1)^\th$ column and the
    $(2p + 1)^\th$ column as the \emph{left-padding} and
    \emph{right-padding} of $v^i_p$, respectively.
    Note that the left-padding of $v^i_1$ is the second column, which is distinct from the left border, and analogously for the right-padding of $v^i_n$ which is distinct from the right border.
    The right-padding and left-padding of consecutive vertices coincide.
    
    \item We delete all the vertical edges of the grid both whose endpoints are \emph{not} in the $(2p)^\th$ column for some $p \in [n]$. 
    In other words, the only vertical paths we keep are those in the columns corresponding to vertices of $V_i$.
    
    \item We subdivide the edges incident to the border $k$ times. 
    Formally, the reduction replaces each such edge with a path with $k$ internal vertices.
    We consider all the new columns obtained in the semi-grid as ``virtual columns'' and consider that this change does not affect the indexing of the columns, for convenience.
    This allows us to denote the column corresponding to vertex $v^i_p$
    as the $(2p)^\th$ column instead of the
    less intuitive notation of ``${(2k + 2p)}^\th$ column''.
    The role of such vertices is to ensure that vertical paths of the grid will stay isometric even after the vertices of the left and right border are made adjacent to the vertices of edge gadgets.

    \item Finally, for each ``cell'' defined by two consecutive rows and two consecutive columns of the obtained grid-like structure, we add a cherry and connect its middle vertex to the four vertices lying at the intersection of these rows and columns.
    We do the same at the left of the first column (of index $2p$) by considering as a cell the two vertices of the column plus the last two subdivided vertices.
    We proceed analogously at the right of the last column.
    The goal of these cherries is to force isometric paths to be either horizontal or (almost) vertical in the obtained semi-grid.
    We call these cherries \emph{grid cherries}; note that there are $(n+1)\cdot (k-2)$ such cherries.
\end{itemize}
This completes the description of the semi-grid $\Gamma_i$ associated to color class $V_i$.

\subsubsection{Encoding Edges}
\label{sec:IPP:tw-w1:construction:edge-gadget}

Consider an edge $(v^i_p, v^j_q)$ of $G$ where $i \neq j \in [k]$ and $p, q \in [n]$.
We start describing the part of the edge gadget that will later be attached to each of $\Gamma_i$ and $\Gamma_j$ and that we call \emph{left and right cables}.
In the following, let us define
\[
    N := 2n^2.
\]
The following steps are better understood accompanied with Figure~\ref{fig:tw-w1:edge-gadget}.

\begin{itemize}
    \item We start by creating a simple path
    $(z_0,x_0, x_1,\dots, x_{\ell})$ that we call \emph{left cable} of $(v^i_p, v^j_q)$ with respect to $\Gamma_i$, 
    where $\ell = N - 2p - k$.
    Note that this value corresponds to $N$ minus the distance of the column of $v^i_p$ from the left border of $\Gamma_i$. 
    We say that $\ell$ is the \emph{length}, $z_0$ is the \emph{core}, and $x_{\ell}$ is the \emph{open end} of the cable. 
    
    \item We create a second simple path $(z'_0,x'_0, x'_1,\dots, x'_{\ell'})$ that we call \emph{right cable} of $(v^i_p, v^j_q)$ with respect to $\Gamma_i$, this time with value $\ell' = N - 2(n-p+1) - k$.
    Note that the value $\ell'$ corresponds to $N$ minus the distance of the column of $v^i_p$ from the right border of $\Gamma_i$.

    \item We create the left cable $(z'''_0,x'''_0, x'''_1,\dots, x'''_{\ell'''})$ and the right cable $(z''_0, x''_0, x''_1,\dots, x''_{\ell''})$ of $(v^i_p, v^j_q)$ with respect to $\Gamma_j$ analogously, with $i$ replaced by $j$, and $p$ replaced by $q$.
\end{itemize}

\begin{figure}[t]
    \centering
    \includegraphics[page=2,width=0.75\textwidth]{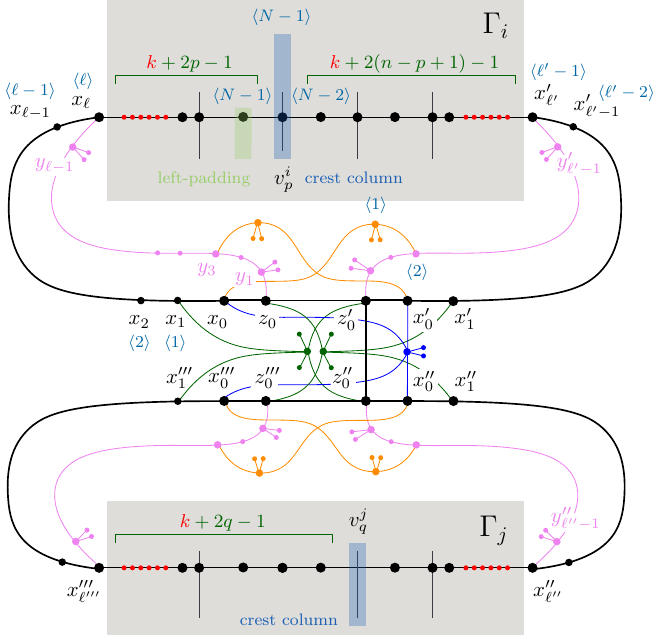}
    \caption{The edge gadget encoding an edge between two color classes. The start cherries are depicted in blue, the core cherries in green, the distance twin-cherries in pink, and the crossing cherries in orange. All these cherries can be assumed to be part of any minimum IP-partition by Lemmas~\ref{lem:cherry-lemma} and \ref{lem:twin-cherry-lemma}, and hence may intuitively be omitted when analyzing the shape of other isometric paths. The distances from $x_0$ are depicted between brackets in blue; note that $x_0$ is equidistant to its crest column and its left-padding.} 
    \label{fig:tw-w1:edge-gadget}
\end{figure}

Note that the left and right cables associated with $\Gamma_i$ may only differ by their length, which is determined by the distance of $v^i_p$ from the left and right border of $\Gamma_i$.
This distance is set so that $x_0$ lies at distance $N-1$ from the left padding of $\smash{v^i_p}$, and analogously for $x'_0$ which lies at distance $N-1$ from the right padding of $v^i_p$.
This $\smash{p^\text{th}}$ column is called the \emph{crest} of $\Gamma_i$ with respect to $\smash{(v^i_p, v^j_q)}$, and the crest of $\Gamma_j$ with respect to $\smash{(v^i_p, v^j_q)}$ is defined analogously.

We connect these four cables by adding the edges $z_0z'_0$, $z'_0z''_0$, and $z''_0z'''_0$ that we call \emph{core edges}.
The obtained path $(z_0, z'_0, z''_0, z'''_0)$ will be called the \emph{core path} in the following.

To these cable and core edges, we add four types of cherries and twin-cherries, defined as follows.
Note that these extra vertices will be assumed to be part of any optimal IP-partition thanks to Lemmas~\ref{lem:cherry-lemma} and \ref{lem:twin-cherry-lemma}, hence, they should be considered as gadgets whose role is to modify the distances in the graph, without changing the ``structure'' of the remaining isometric paths.

\begin{itemize}
    \item First we add a \emph{start cherry}, in blue in Figure~\ref{fig:tw-w1:edge-gadget}, whose middle vertex is adjacent to $x_0, x'_0, x''_0$, and $x'''_0$.
    The role of this cherry will be to simplify the analysis by ensuring these adjacent vertices to lie in distinct isometric paths of a minimum IP-partition.

    \item Then we add \emph{core cherries}, in green in Figure~\ref{fig:tw-w1:edge-gadget}, as follows.
    The \emph{left} one has its middle vertex adjacent to the core vertices $z'_0, z''_0$ as well as to $x_1$ and $x'''_1$.
    The \emph{right} one has its middle vertex adjacent to the core vertices $z_0, z'''_0$ as well as to $x'_1$ and $x''_1$.
    The role of these cherries will be to force that cables cannot be isometrically extended using the core of another cable.

    \item Then, we add a twin-cherry as follows: we create a path $y_1,\dots, y_{\ell-1}$ with $y_1$ adjacent to $z_0$, $y_{\ell-1}$ adjacent to $x_{\ell}$, and add two leaves to each of $y_1$ and $y_{\ell-1}$.
    We call this twin-cherry the \emph{core twin-cherry} of the left path of $\smash{(v^i_p, v^j_q)}$ with respect to $\Gamma_i$.
    This is depicted in pink in Figure~\ref{fig:tw-w1:edge-gadget}.
    The role of this twin-cherry is to make $z_0$ equidistant to $x_{\ell-1},x_{\ell}$.
    We define the twin-cherry of the right path of $\smash{(v^i_p, v^j_q)}$ with respect to $\Gamma_i$ analogously, and do the same for the left and right paths of $\smash{(v^i_p, v^j_q)}$ with respect to $\Gamma_j$.

    \item Finally, we add a \emph{crossing cherry}, in orange in Figure~\ref{fig:tw-w1:edge-gadget}, whose middle vertex is adjacent to $y_3$ and $y'_3$, and call it the crossing cherry of $\smash{(v^i_p, v^j_q)}$ with respect to $\Gamma_i$.
    We define the crossing cherry associated to  $\smash{(v^i_p, v^j_q)}$ with respect to $\Gamma_j$ analogously.
    The role of these cherries is to make $x_0$ at distance $N-1$ from its associated crest column, so that it becomes equidistant to both this column and its left-padding.
\end{itemize}

This concludes the construction of an edge gadget for
edge $(v^i_p, v^j_q)$.
We connect it to the appropriate semi-grids by identifying open ends of the paths with the vertices in the borders of the semi-grids
$\Gamma_i$ or $\Gamma_j$, depending on whether they are left paths, or right paths.
More formally, we identify $x_{\ell}$ with vertex $a^i_{j}$, and $x'_{\ell'}$ with vertex $b^i_{j}$.
Then we do the same for $\Gamma_j$ by identifying $\smash{x'''_{\ell'''}}$ with vertex $\smash{a^j_{i}}$, and $\smash{x''_{\ell''}}$ with vertex $\smash{b^j_{i}}$. 
Figure~\ref{fig:tw-w1:edge-gadget} illustrates the complete picture.

Let us stress the fact that one such edge gadget is added for each single edge $\smash{(v^i_p, v^j_q)}$ between $V_i$ and $V_j$, and that all such gadgets are thus attached to the same $\smash{j^\th}$ row of $\Gamma_i$, and to the same $\smash{i^\th}$ row of $\Gamma_j$, via their open ends.
In particular, open ends separate each edge gadget from the other gadgets and the associated semi-grid.
In the following, we call \emph{inner vertices} of an edge gadget the vertices of the edge gadget that are distinct from its open ends.

\subsubsection{Valve Cherries}
\label{sec:IPP:tw-w1:construction:valves}

For each semi-grid $\Gamma_i$, we add two \emph{valve cherries}: one on the left of the semi-grid, and one on the right. See Figure~\ref{fig:tw-w1:valves}
for an illustration.
The middle vertex of the left cherry is made adjacent to $x_{\ell-2}$ and $x_{\ell-1}$ in every edge gadget attached to the left border of the semi-grid. 
The middle vertex of the right cherry is made adjacent to $x'_{\ell-2}$ and $x'_{\ell-1}$ in every edge gadget attached to the right border.
We will argue that the valve cherries prevent isometric paths to intersect distinct edge gadgets under additional assumptions. 

\begin{figure}[ht]
    \centering
    \includegraphics[page=5,width=0.5\textwidth]{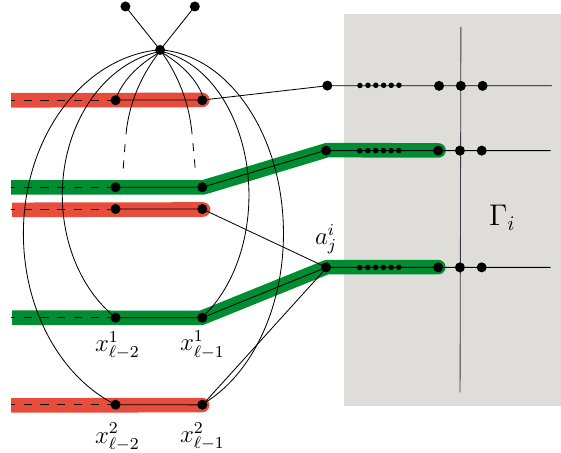}
    \caption{Valve cherries attached to the edge gadgets attached to a border of a semi-grid, here the left border. With the notation $\smash{x^1_{\ell-2}}$ we mean the vertex $x_{\ell-2}$ in the edge gadget associated to edge $e_1$.}
    \label{fig:tw-w1:valves}
    \vspace{-0.5cm}
\end{figure}

\subsubsection{Solution Size}\label{sec:IPP:tw-w1:solution-size}

The reduction sets
\[
    k' = 
    k\cdot (n+1)\cdot (k-2) + 23 \cdot |E(G)| + \binom{k}{2} + k + 2k
\] 
and returns $(H, k')$ as an instance of \SoPP. 

We present an informal justification for the above value of $k'$ which may help
the reader in the next section where we prove properties of IP-partitions of $H$.
The first additive term corresponds to the number of grid cherries that are found in the semi-grids.
The second term corresponds to the number of paths used
to partition the vertices in edge gadgets, 19 of which will consist of cherries and twin-cherries.
The third term corresponds to the number of selected edges: here, the partition needs to spend an extra isometric path for each selection.
A path from the selected edges can be extended to partition almost all the vertices in the semi-grid except for their crest.
If selected edges correspond to a multicolored clique
in $G$, then the solution can partition the remaining crest
vertices in semi-grids using $k$ vertical isometric paths, which corresponds to the fourth additive term.
The last additive term corresponds to the valve cherries.

\subsection{Properties of IP-partitions}\label{sec:IPP:tw-w1:properties}

Let $\calP$ be an IP-partition of $H$ of minimum cardinality.
Note that by Lemmas~\ref{lem:cherry-lemma} and \ref{lem:twin-cherry-lemma} we can assume without loss of generality that $\calP$ contains, for each edge gadget:
\begin{itemize}
    \item one start cherry, as represented in blue in Figure~\ref{fig:tw-w1:edge-gadget};
    \item two core cherries, as represented in green in Figure~\ref{fig:tw-w1:edge-gadget};
    \item four twin-cherries, made of three paths each as represented in pink in Figure~\ref{fig:tw-w1:edge-gadget}; and
    \item four crossing cherries, as represented in orange in Figure~\ref{fig:tw-w1:edge-gadget};
\end{itemize}
to which we add the $k\cdot (n+1)\cdot (k-2)$ grid cherries as represented in Figure~\ref{fig:tw-w1:semi-grid}, and the $2k$ valve cherries represented in Figure~\ref{fig:tw-w1:valves}.
This sums up to a total number of
\[
    k\cdot (n+1)\cdot (k-2) + 19 \cdot |E(G)| + 2k
\]
cherries.
Let us denote $\calQ\subseteq \calP$ the set of remaining paths (which are not among the cherries described above) in $\calP$.
We now focus on characterizing this set.
In the incoming argument, it is crucial to keep in mind that the paths of $\calQ$ may not contain the vertex of a cherry of $\calP$.

We start by characterizing the intersection of $\calQ$ with semi-grids.
In the statement below, by \emph{horizontal} we mean a subpath of a single row, and by \emph{vertical} we mean a subpath of a single column; hence cherries and left- and right-paddings are excluded.

\begin{claim}\label{claim:tw-w1:vertical-is-isometric}
    The vertical paths in the semi-grids are isometric in $H$.
\end{claim}

\begin{proof}
    The statement is clear if we consider the subgraph of $H$ induced by the vertices of a semi-grid.
    Now note that because of the $k$ subdivided vertices, in red in Figure~\ref{fig:tw-w1:semi-grid}, with $k$ also corresponding to the maximum length of a column, an isometric path cannot have its endpoints in a column of the grid, and its inner vertices outside.
    This concludes the proof. 
\end{proof}

\begin{claim}\label{claim:tw-w1:path-in-grids-horizontal-or-vertical}
    The intersection of each path of $\calQ$ with the grid is either horizontal, vertical, or it consists of a vertical path except possibly for its endpoints lying on the padding of a column.
\end{claim}

\begin{proof}
    Consider a semi-grid, and let us call \emph{crossing} any vertex that either lies both on a row and on a column of the semi-grid, or that is equal to a subdivision vertex adjacent to a padding.
    Note that these vertices are precisely those that are adjacent to the middle vertex of a grid cherry; see Figure~\ref{fig:tw-w1:semi-grid}.
    Let $P$ be the intersection of an isometric path of $\calQ$ with the semi-grid, and suppose it is neither horizontal nor vertical.
    Suppose, towards a contradiction, that $P$ contains two non-aligned crossings (i.e., crossings that do not share the same row or column), and consider two such vertices that are consecutive (ignoring non-crossing vertices) in $P$.
    Note that the unique isometric path between these crossings goes through the grid cherry, a contradiction.
    Thus, $P$ only contains aligned crossings, and we derive that $P$ may contain at most two vertices lying on paddings, being its endpoints. 
\end{proof}

We say that a path $Q$ of $H$ \emph{extends} a path $P$ if it contains it as a proper subpath; $Q$ \emph{isometrically extends} a path $P$ if in addition $Q$ is isometric in $H$.
The goal of the remainder of the section is to characterize different ways for the cables of edge gadgets to be covered by an IP-partition of $H$, and to argue that one of this ways can be extended, while the other cannot.
Towards this, we introduce the following terminology.

\begin{definition}\label{def:extendable-non-extendable-paths}
    We call the path $(x_0, \dots, x_{\ell})$ of an edge gadget and any of its extensions \emph{extendable}, and the path $(z_0, x_0, \dots, x_{\ell - 1})$ is \emph{non-extendable}; this definition is analogously extended to the four cables of an edge gadget.
\end{definition}

These paths are depicted in green and red in Figures~\ref{fig:tw-w1:extendable-paths} and \ref{fig:tw-w1:non-extendable-paths}, respectively. 
In the following, for simplicity, we will focus on left cables and thus denote the vertices of cables without primes in the notation, e.g., as $(z_0, x_0, x_1, \dots, x_{\ell})$; we stress the fact that all the statements and arguments hold for right cables by symmetry.

Note that extendable and non-extendable paths are isometric in $H$.
We justify these names with the following claims. 

\begin{figure}[t]
    \centering
    \includegraphics[page=3,width=0.75\textwidth]{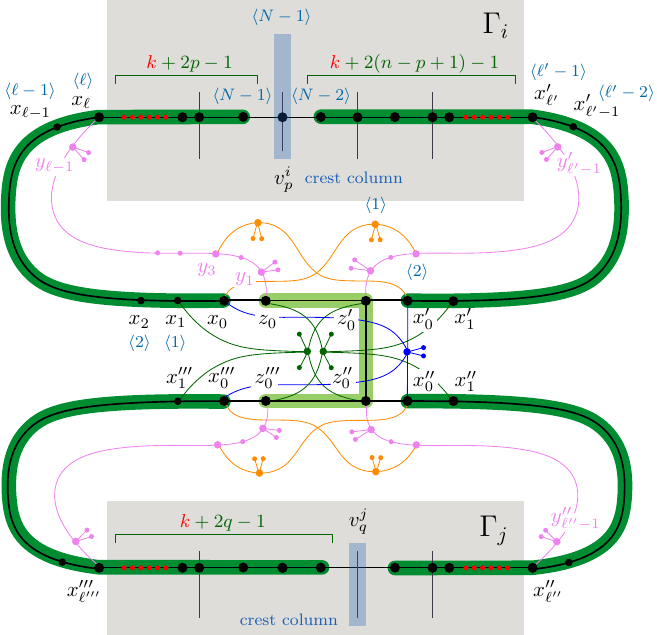}
    \caption{Four extendable paths and one core path partitioning all the vertices of an edge gadget, but those in cherries and twin-cherries. The distances from $x_0$ are depicted between brackets in blue; note that $x_0$ is equidistant to its crest column and its left-padding. Hence, the path containing of $x_0$ cannot be extended to cover the crest column.} 
    \label{fig:tw-w1:extendable-paths}
\end{figure}

\begin{figure}[t]
    \centering
    \includegraphics[page=4,width=0.75\textwidth]{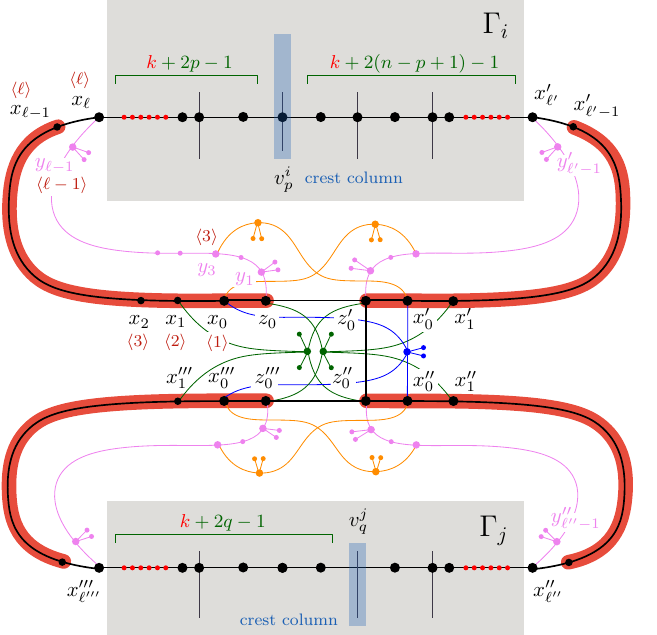}
    \caption{Four non-extendable paths partitioning the vertices that are neither open ends, nor part of a cherry or of a twin-cherry. 
    The distances from $z_0$ are depicted between brackets in red; note that $z_0$ is equidistant to $x_{\ell-1},x_{\ell}$, hence that its path cannot be extended in the semi-grid.}
    \label{fig:tw-w1:non-extendable-paths}
\end{figure}

\begin{claim}\label{claim:tw-w1:extendable-property}
    The extendable paths can only be isometrically extended in the semi-grid, and this can be done up to the associated crest column excluded, i.e., their extension can contain the paddings of the crest column but not the crest column.
\end{claim}

\begin{claimproof}
    Consider a pair of left and right cables as in Figure~\ref{fig:tw-w1:extendable-paths}.
    Note that any isometric extension of the path $(x_0, \dots, x_{\ell})$ has $x_0$ as an endpoint, since $x_{0}$ and $z_{0}$ are equidistant from $x_\ell$, due to the twin-cherry, represented in pink in Figure~\ref{fig:tw-w1:extendable-paths}. 
    Now, by the choice of $\ell$ and $\ell'$ in the definition of the cables, together with the distances induced by the crossing cherries represented in orange in Figure~\ref{fig:tw-w1:extendable-paths}, both the crest column and its left padding lie at distance $N-1$ from $x_0$.
    Hence, $(x_0, \dots, x_{\ell})$ cannot be extended to contain the vertex of the crest column.
    However, it is easily checked that it can be extended up to the left padding of this crest column, as desired.
\end{claimproof}

\begin{claim}\label{claim:tw-w1:non-extendable-property}
    The non-extendable paths cannot be isometrically extended.
\end{claim}

\begin{claimproof}
    See Figure~\ref{fig:tw-w1:non-extendable-paths}.
    The fact that the path $(z_0, x_0, \dots, x_{\ell - 1})$ cannot be extended follows from the fact that $x_{\ell-1}$ and $x_{\ell}$ are equidistant from $z_0$, due to the twin-cherry represented in pink in Figure~\ref{fig:tw-w1:extendable-paths}, and the fact that $z_{0}$ and $z'_{0}$ are equidistant from $x_1$, due to the core cherries, represented in green in Figure~\ref{fig:tw-w1:non-extendable-paths}.
\end{claimproof}

We derive the following as a corollary of Claim~\ref{claim:tw-w1:non-extendable-property}, noting that if an isometric path contains all the vertices of a cable except for its open end, then it contains a non-extendable path.

\begin{corollary}\label{cor:tw-w1:inner-1-path-non-extendable}
    If one path in $\calQ$ covers all the vertices of a cable except for its open end, then this path is a non-extendable path.
    In particular, at least two paths in $\calQ$ intersect the vertices of each cable.
\end{corollary}

We are now interested in expressing bounds on the number of paths needed to partition our edge gadgets. 

\begin{claim}\label{claim:tw-w1:start-cherry}
    Each of the four neighbors $x_0,\, x'_0,\, x''_0$, and $x'''_0$ of the middle vertex of a start cherry belongs to a distinct path in $\calQ$.
\end{claim}

\begin{proof}
    Consider the start cherry represented in blue in Figures~\ref{fig:tw-w1:edge-gadget}--\ref{fig:tw-w1:non-extendable-paths}.
    Note that the only isometric path between any two such neighbors goes through a cherry.
    Hence, no two distinct paths of $\calQ$ may contain two such neighbors.
\end{proof}

Recall that the inner vertices of an edge gadget are those vertices of the edge gadget that are distinct from its open ends.

\begin{claim}\label{claim:tw-w1:edge-at-most-4-non-extendable}
    If at most four paths in $\calQ$ cover all the inner vertices of an edge gadget, then these paths are precisely the non-extendable paths.
\end{claim}

\begin{proof}
    By Claim~\ref{claim:tw-w1:start-cherry}, at least four distinct paths of $\calQ$ intersect the inner vertices of the gadget, precisely on distinct neighbors of the middle vertex of the start cherry, represented in blue in Figures~\ref{fig:tw-w1:edge-gadget}--\ref{fig:tw-w1:non-extendable-paths}.
    Thus, these paths are those that cover the full gadget.
    Note that each of these paths cannot visit another cable, due to the core cherries, represented in green in Figures~\ref{fig:tw-w1:edge-gadget}--\ref{fig:tw-w1:non-extendable-paths}.
    We conclude by Corollary~\ref{cor:tw-w1:inner-1-path-non-extendable}.
\end{proof}

\begin{claim}\label{claim:tw-w1:no-path-between-two-gadgets}
    It can be assumed that no path in $\calQ$ intersects the inner vertices of two distinct edge gadgets.
\end{claim}

\begin{proof}
    Suppose that such a path $P\in \calQ$ intersecting the inner vertices of two distinct edge gadgets exists.
    Then, it must be connecting two gadgets associated with distinct edges $e_1,e_2$ via their open end.
    By the distances induced by the valve cherry, it must be that these gadgets share a same open end. 
    Moreover, $P$ contains one neighbor of the open end in each of the two edge gadgets, since it contains an inner vertex.
    Let us assume without loss of generality that the open end lies on the left border of a grid, and denote the vertices of the edge gadget associated with $e_1$ and $e_2$ with $1$ and $2$ in superscript, as in Figure~\ref{fig:tw-w1:valves}.
    
    Note that $P$ cannot reach $x^1_{\ell-2}$ and $x^2_{\ell-2}$, since the unique shortest path between these vertices goes through the valve cherry; see Figure~\ref{fig:tw-w1:valves}.
    Similarly, it cannot contain both $x^1_{\ell-1}$ and $x^2_{\ell-2}$ at the same time, nor $x^1_{\ell-2}$ and $x^2_{\ell-1}$ at the same time.
    Hence, $P$ consists of a path of length~2 induced by $x^1_{\ell-1}$, $x^1_{\ell-1}$ and their mutual neighbor, which is the open end.
    
    Now, by Claim~\ref{claim:tw-w1:start-cherry}, we also know that there is a path $Q\in \calQ$ containing the vertex $x^1_0$.
    Consider the path $Q'$ containing $x^1_{\ell-2}$.
    Note that this path either equals $Q$, or has its vertices in $\{\smash{x^1_1,\dots,x^1_{\ell-2}}\}$.
    In any case, it is easily seen that this path can be extended to contain $x_{\ell-1}$, and $P$ can be reduced, so that it does not intersect the edge gadget associated with $e_1$ anymore. 
    Repeating this argument, we may assume that in $\calQ$, no path intersects the inner vertices of two distinct edge gadgets, as desired.
\end{proof}

We derive the following as a corollary of Claims~\ref{claim:tw-w1:start-cherry} and~\ref{claim:tw-w1:no-path-between-two-gadgets}.

\begin{corollary}\label{cor:edge-require-at-least-4-paths}
    There are at least four distinct paths in $\calQ$ intersecting each edge gadget, and distinct edge gadgets are intersected by distinct such paths.
\end{corollary}

We are now interested in covering all the vertices of an edge gadget, including its open ends (which are shared by its associated grids and other edge gadgets).

\begin{claim}\label{claim:tw-w1:edge-at-most-5-extendable}
    If at most five paths in $\calQ$ cover all the vertices of an edge gadget, including its open end, then these paths are precisely the extendable paths of each cable, together with a single core path. Moreover, such a family can be extended in order to cover their associated rows, except for their crest columns.
\end{claim}

\begin{proof}
    By Corollary~\ref{cor:edge-require-at-least-4-paths}, at least four distinct paths of $\calQ$ intersect the edge gadget, namely on $x_0,x'_0, x''_0$ and $x'''_0$.
    As we only have one extra path to cover the edge gadget, including the open ends, we derive that at least three of the above paths are of the extendable type.
    Thus, the three remaining core vertices and the remaining cable are covered by two other paths.
    Note that the path that contains the remaining open end must be the one containing a neighbor of the start cherry among $\{x_0,x'_0, x''_0,x'''_0\}$, as a single path cannot contain the four core vertices plus one of these vertices.
    Hence, this path is of the extendable type as well, and the last path is a core path.
    The last statement follows by Claim~\ref{claim:tw-w1:extendable-property}.
\end{proof}

\begin{claim}\label{claim:tw-w1:no-more-than-5}
    It can be assumed that no more than five paths in $\calQ$ intersect an edge gadget.
\end{claim}

\begin{proof}
    By Claim~\ref{claim:tw-w1:edge-at-most-5-extendable}, five paths suffice to cover an edge gadget, and moreover, paths covering the open end can be extended into the grid up to the associated crest column.
    Thus, the only utility for intersecting a gadget with strictly more than five paths would be to have such an additional path to cover more in the semi-grid than the five paths provided by Claim~\ref{claim:tw-w1:edge-at-most-5-extendable} can.
    However, by Claim~\ref{claim:tw-w1:path-in-grids-horizontal-or-vertical}, this extra path must be horizontal, hence cannot be used to cover other rows in the semi-grid.
    We derive that the same scenario is achieved by having the extendable paths provided by Claim~\ref{claim:tw-w1:edge-at-most-5-extendable} to be extended to the crest column, and the remaining vertex of the crest column to be covered by one extra one-vertex path.
    Hence, we can assume that no more than five paths in $\calQ$ intersect a given edge gadget, as desired.
\end{proof}

We are now interested in the consequences of using extendable or non-extendable paths to cover edge gadgets.

\begin{definition}
    We say that an edge $(v^i_p, v^j_q)$ in $G$ is selected by $\calQ$ if the inner-vertices of the edge gadget encoding it in $H$ intersect at least five paths in $\calQ$; otherwise, we say that the edge is not selected by $\calP$.
\end{definition}

\begin{claim}\label{claim:tw-w1:at-most-one-selection}
    No two edge gadgets attached to the same row of a semi-grid are selected by $\calQ$.
\end{claim}

\begin{proof}
    Note that edge gadgets attached to the same row are separated by their mutual open end.
    Moreover, if an edge gadget is selected by $\calQ$, we may assume that its open end is intersected by one of its five paths, as otherwise by Claim~\ref{claim:tw-w1:edge-at-most-4-non-extendable} we may cover the gadget with four paths and modify $\calQ$ accordingly.
    Let us suppose that at least two edge gadgets are selected by $\calQ$ and make a case analysis depending on the way their open ends are covered.

    Let us first assume that one such gadget has both its open ends covered. 
    Recall that by Claim~\ref{claim:tw-w1:no-path-between-two-gadgets}, the paths of $\calQ$ do not intersect several edge gadgets.
    Thus, and since the open ends separate all other edge gadgets from the graph, we conclude that the paths of $\calQ$ covering the other gadgets attached to the same row only cover their inner vertices.
    Again, we can modify $\calQ$ so that these gadgets are covered by four paths according to Claim~\ref{claim:tw-w1:edge-at-most-4-non-extendable}, and hence that the gadgets are no longer selected by $\calQ$.

    Let us now suppose that an edge gadget is selected, but only one of its two open ends are covered using these paths.
    Thus, the open end that is not covered by these paths must be covered by another path $Q$ from $\calQ$.
    By Claim~\ref{claim:tw-w1:edge-at-most-5-extendable}, we can modify $\calQ$ so that both its open ends are covered by at most five paths, and reduce $Q$ accordingly.
    We are reduced to the previous case, for which we proved that $\calQ$ can be further modified to satisfy the statement.
    This concludes the case study, hence the proof.
\end{proof}

\begin{claim}\label{claim:tw-w1:no-selection-implies-confined-additional-row}
    If no edge between $V_i$ and $V_j$ is selected by $\calQ$, then at least one additional path in $\calQ$ is contained in the associated rows in each of $V_i$ and $V_j$.
    Hence, for any two distinct pairs of color classes for which no edge has been selected by $\calQ$, such paths are distinct.
\end{claim}

\begin{proof}
    Consider a pair $V_i,V_j$ such that no edge between $V_i$ and $V_j$ is selected by $\calQ$. 
    By Claim~\ref{claim:tw-w1:edge-at-most-4-non-extendable}, the paths of $\calQ$ intersecting the edge gadgets associated with $V_i,V_j$ are all of the non-extendable type.
    Let us focus on $\Gamma_i$ and consider the open ends $a:=a^i_j$ and $b:=b^i_j$.
    Note that their neighbor in the edge gadgets are covered by the non-extendable paths described above.
    Their other neighbor is unique and consists of a subdivision vertex in the semi-grid.
    Hence, there is a path in $\calQ$ starting at $a$.
    By Claim~\ref{claim:tw-w1:path-in-grids-horizontal-or-vertical}, this path must be either horizontal, vertical, or almost vertical with possibly the endpoints not lying on a column.
    However, this path is of the horizontal type.
    Hence, it is confined in the row until it reaches $b$.
    At that point, it may not continue since its neighbors are covered by the non-extendable paths described above.
    This concludes the proof.
\end{proof}

\subsection{Proof of Theorem~\ref{thm:IPP:w1-tw-hard}}
\label{sec:IPP:tw-w1:proof}

We are now ready to give the proof of Theorem~\ref{thm:IPP:w1-tw-hard}, that we split into three lemmas.
Let us recall that the reduction sets 
\[
    k' := k\cdot (n+1)\cdot (k-2) + 23 \cdot |E(G)| + \binom{k}{2} + k + 2k
\]
given an instance $(G, k)$ of \textsc{Multicolored Clique} where $n$ is the size of color classes.

\begin{lemma}\label{lemma:tw-w1:hard-backward}
    If $(H, k')$ is a \yes-instance of \SoPP, then $(G, k)$ is a \yes-instance of \textsc{Multicolored Clique}.
\end{lemma}

\begin{proof}
    Consider an IP-partition $\calP$ of $H$ of minimum cardinality.
    As discussed in the beginning of Section~\ref{sec:IPP:tw-w1:properties}, by Lemmas~\ref{lem:cherry-lemma} and \ref{lem:twin-cherry-lemma}, we may assume that $\calP$ contains the
    \[
        k\cdot (n+1)\cdot (k-2) + 19\cdot |E(G)| + 2k
    \]
    cherries defined in Section~\ref{sec:IPP:tw-w1:construction} and represented in Figures~\ref{fig:tw-w1:semi-grid}--\ref{fig:tw-w1:non-extendable-paths}.
    Let $\calQ$ be the set of remaining paths in $\calP$.
    By Corollary~\ref{cor:edge-require-at-least-4-paths}, at least $4\cdot |E(G)|$ paths in $\calQ$ intersect the inner vertices of edge gadgets.
    Thus, the remaining budget to cover semi-grids is of 
    \[
        \binom{k}{2} + k.
    \]
    By Claim~\ref{claim:tw-w1:no-selection-implies-confined-additional-row}, if in $\calQ$ all the edge gadgets associated with a pair of color classes $V_i,V_j$ are intersected by exactly four paths, then an additional path is needed to cover the associated row in $V_i$, and may only be used to cover this row; the same holds for $V_j$.
    This sums up to $|k-1|\cdot k$ additional paths to cover the rows of each semi-grid, which is over budget for $k\geq 4$.
    
    The other option for an edge of $G$ between $V_i$ and $V_j$ is to be selected by $\calQ$.
    In that case, by Claim~\ref{claim:tw-w1:no-more-than-5}, its corresponding gadget is intersected by exactly five paths of $\calQ$, and by Claim~\ref{claim:tw-w1:at-most-one-selection}, it is the only one among the gadgets associated to $V_i,V_j$.
    Moreover, by Claim~\ref{claim:tw-w1:edge-at-most-5-extendable}, selecting such an edge allows to cover the rows of the two associated semi-grids, except for their crest column which require an additional path.
    Note that, unless exactly one additional path per grid is used to cover its crest vertices, this option is also over budget, with a total of at least 
    \[
        \binom{k}{2} + k + 1.
    \]
    
    Thus, the only possibility for $\calQ$ to be within budget is to have exactly one additional path per grid that is used to cover its crest vertices.
    However, to have a single isometric path to cover the crest vertices of a given grid, by Claim~\ref{claim:tw-w1:path-in-grids-horizontal-or-vertical}, it must be that these vertices are part of a same column of the semi-grid.
    In other words, for any color class $V_i$, and any other color class $V_j$, there is a selection of one edge between $V_i$ and $V_j$ such that these edges all share the same crest column of index $2p$ in $V_i$.
    By construction, all these edges share $v^i_p$ for their endpoint.
    Hence, the set of such $v^i_p$'s over all color classes forms a clique in $G$.
    This concludes the proof.
\end{proof}

\begin{lemma}\label{lemma:tw-w1:hard-forward}
    If $(G, k)$ is a \yes-instance of \textsc{Multicolored Clique}, then $(H, k')$ is a \yes-instance of \SoPP.
\end{lemma}

\begin{proof}
    Consider a clique $K$ of $(G, k)$.
    We construct an IP-partition of $H$ as follows.
    First, we consider the set of $k\cdot (n+1)\cdot (k-2) + 19\cdot |E(G)| + 2k$ cherries defined in Section~\ref{sec:IPP:tw-w1:construction}.
    Then, we add the family of disjoint isometric paths obtained by selecting each edge in $K$, i.e., this family consists of the extendable paths as defined in Definition~\ref{def:extendable-non-extendable-paths}.
    By Claim~\ref{claim:tw-w1:edge-at-most-5-extendable}, these paths suffice to cover the corresponding gadgets and can be extended to cover their associated rows up to their crest column, which is left uncovered.
    Since all edges having an endpoint in a color class $V_i$ share the same endpoint, the crest columns in $\Gamma_i$ of each selected edge coincide.
    We add one more vertical path to cover the full grid.
    This adds a total of $5\cdot \binom{k}{2}+k$ paths.
    For every other edge, we consider the four non-extendable paths as defined in~\ref{def:extendable-non-extendable-paths}.
    By Claim~\ref{claim:tw-w1:edge-at-most-4-non-extendable}, four paths per such edge gadget indeed suffice.
    This adds a total of $4\cdot (|E(G)|-\binom{k}{2})$ paths, and completes the construction of the family.
    Note that this family precisely contains $k'$ paths.
    By construction, all paths are disjoint, and they cover the full graph.
    The fact that they are isometric is trivial for cherries, follows from Claim~\ref{claim:tw-w1:edge-at-most-5-extendable} and Claim~\ref{claim:tw-w1:edge-at-most-4-non-extendable} for extendable and non-extendable paths, and from Claim~\ref{claim:tw-w1:vertical-is-isometric} for vertical paths in the semi-grids.
    This concludes the proof.
\end{proof}

\begin{lemma}\label{lemma:tw-w1:hard-pathwidth}
    The graph $H$ has pathwidth $O(k^2)$.
\end{lemma}

\begin{proof}
    Consider the set of vertices $X\subseteq V(H)$
    which consists of the union of the left and right borders in the semi-grids, together with the middle vertex of each valve cherry.
    By construction, $|X| \in O(k^2)$, and $H - X$
    is a disconnected graph whose connected components are
    (1) subgraphs of the $(2k + 2n + 1) \times (k-1)$ grid with a constant number of edges and vertices in each cell,
    (2) isolated endpoints of valve cherries, and
    (3) collections of eight long paths with a constant number
    of additional edges. 
    Each of these components has pathwidth
    $O(k)$, and hence the pathwidth of $H$ is $O(k^2)$.   
\end{proof}

We conclude with Theorem~\ref{thm:IPP:w1-tw-hard}, that we restate here, as a corollary of Lemmas~\ref{lemma:tw-w1:hard-backward}, \ref{lemma:tw-w1:hard-forward}, and \ref{lemma:tw-w1:hard-pathwidth}, noting that the graph $(H,k')$ can be computed in polynomial time given an instance of \textsc{Multicolored Clique}.

\restateIPPtwhard*

\newcommand{\cycle}[1]{A_{#1}}
\newcommand{\clGroup}[1]{C'_{#1}}
\newcommand{\varV}[2]{v_{#1}^{#2}}
\newcommand{\clGroupName}[2]{c_{#1}^{#2}}
\newcommand{\cldummy}[2]{d_{#1}^{#2}}
\newcommand{\fnb}[1]{\beta\mathopen{}\left(#1\mathclose{}\right)} 
\newcommand{\distanceMod}[4]{D\mathopen{}\left(#1,#2,#3,#4\mathclose{}\right)}

\section{Lower Bound w.r.t.~Pathwidth and Diameter}
\label{sec:IPP:eth-diam-tw}

In this section, we prove that \SoPP{}
does not admit an algorithm running in time $O(\diam^{o(\tw^2/\log^3(\tw))})$,
unless the \textsf{Randomized ETH} fails.
Towards that, we present a reduction from \textsc{Sparse 3-SAT} to 
\SoPP. 
\textsc{Sparse 3-SAT} has been introduced by Gourvès et al.~\cite{gourves2024filling} as a sparse variation of 3-SAT.
{We consider the following slight variation of the problem and argue that this modification can indeed be considered without loss of generality.} 

\medskip

\defproblem{\textsc{Sparse 3-SAT}}{An integer $n$ which is a 
perfect square, and a $3$-SAT formula with at most 
$n$ variables and at most $n$ clauses such that each variable appears 
in at most $3$ clauses.
Moreover, a partition $\{V_1,\ldots,V_{\sqrt{n}}\}$ of the set of 
variables $V$ and a 
partition $\{C_1,\ldots,C_{\sqrt{n}}\}$ of the set of clauses $C$ 
such that each part is of size at most $\sqrt{n}$ and
{for every}
$i,j\in \left[\sqrt{n}\right]$ the cardinality of the set 
$\{(x,c) : x\in V_i,\, c\in C_j,\, x\in c\}$ is at most one.}{Does there exist a satisfying assignment of the formula?}

\medskip

In the original definition of the problem~\cite[Definition 2]{gourves2024filling} it is mentioned that 
``the number of variables 
of $V_i$ which appear in at least one clause of $C_j$ is at most one.''
This does not forbid a variable $x\in V_i$ to appear in multiple clauses in $C_j$.
However, we do \emph{not} want to allow this and need this
stronger restriction, as stated in the problem definition above.
As evident from \cite[Lemma~11]{gourves2024filling},
the conditional lower bound mentioned below also holds
for the version of the problem that we state.

\begin{proposition}[{\cite{gourves2024filling}}]
\label{prop:sparse-hard}
Unless the \textsf{Randomized ETH} fails,
\textsc{Sparse 3-SAT} does not admit an algorithm running 
in time $2^{o(n)}$.
\end{proposition}

\begin{figure}[ht]
    \centering
    \includegraphics[page=1,width=0.75\linewidth]{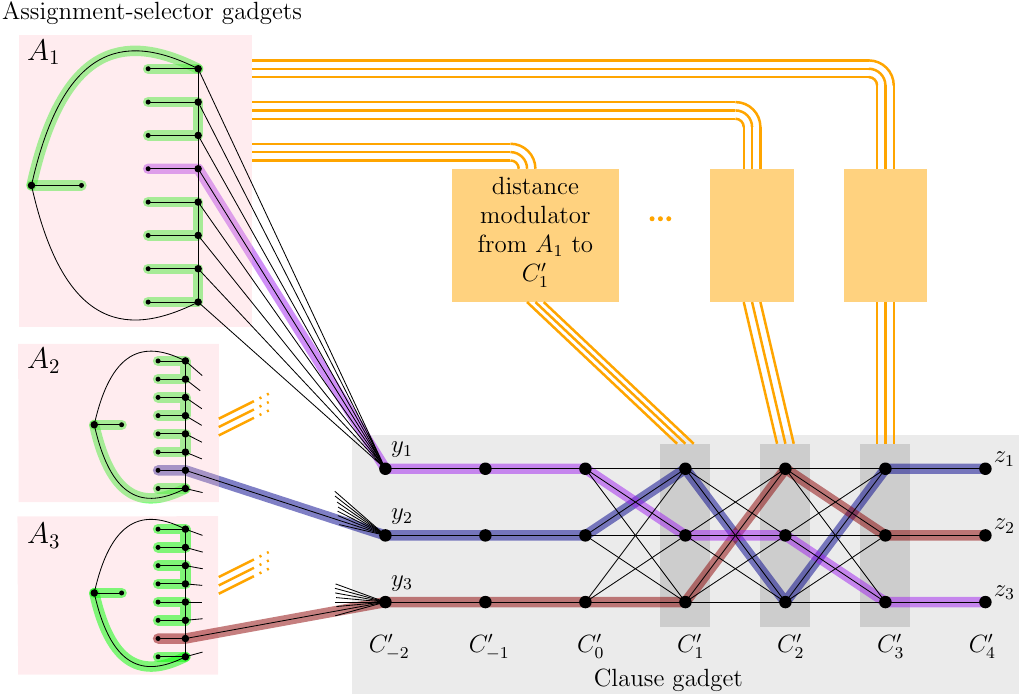}
    \caption{A schematic diagram of the reduction where $\sqrt{n}=3$. The graph consists in three parts: $\sqrt{n}$ assignment-selector gadgets (pink), a clause gadget (gray) and $\sqrt{n}\times\sqrt{n}$ distance modulators (orange). In any optimal IP-partition, only one path (purple, red, blue) -- called assignment-selection paths -- can leave each assignment-selector gadget and the location of their endpoint there encodes an assignment of the corresponding variables. These assignment-selection paths then traverse the clause gadget until they reach their other endpoints on $C'_4$. On the way, they (hopefully) cover vertices on sets $C'_1,\dots, C'_{\sqrt{n}}$ that represent the sets of clauses. To prevent an assignment-selection path to cover a clause vertex that is not satisfied by the corresponding assignment, we use distance modulators (orange) that shorten the distance by one, thus preventing an isometric path to contain both vertices.}
    \label{fig:schematic}
\end{figure}

\subsection{Overview of the Reduction}

The reduction takes as input an instance 
$(\phi, \{V_1, \dots, V_{\sqrt{n}}\}, \{C_1, \dots, C_{\sqrt{n}}\})$
of \textsc{Sparse 3-SAT},
runs in time $2^{O(\sqrt{n})}$, and returns an instance
$(G_{\phi}, k)$ of \SoPP where $k=2^{O(\sqrt{n})}$. The treewidth and diameter of $G_{\phi}$ are $O(\sqrt{n}\log n)$ and $O(\sqrt{n})$, respectively. Altogether, these bounds imply that \SoPP does not admit an $O(\diam^{o(\tw^2/\log^3(\tw))})$-time algorithm, unless the \textsf{Randomized ETH} fails.

The graph of the reduction consists of three parts: the \emph{assignment gadget} that encodes the possible assignments of the variables, a \emph{clause gadget} that contains a vertex for each clause,
and \emph{distance modulator} gadgets that encode the formula. 

For each variable group $V_i$, the \emph{assignment-selector gadget} $A_i$ consists of vertices corresponding to the $2^{\sqrt{n}}$
possible (partial) assignments of the 
variables in this group.  See 
\Cref{sec:IPP:eth-diam-tw:construction:assignments} for details. It is designed in a way that only one path can leave the gadget. Altogether, these $\sqrt{n}$ paths, called \emph{assignment-selection paths}, encode an assignment of the variables.

The clause gadget consists of a chain of $\sqrt{n}$ 
sets $\clGroup{1},\ldots,\clGroup{\sqrt{n}}$ each representing a clause set, together with four additional sets $\clGroup{-2}, \clGroup{-1}, \clGroup{0}$ and $\clGroup{\sqrt{n}+1}$ that will force some properties on isometric paths. 
Each of these sets consists of $\sqrt{n}$ vertices, and each clause of the formula is associated to one of these vertices. 
The assignment gadget is connected to the clause gadget in a way that the assignment-selection paths need to traverse all the sets $\clGroup{1},\ldots,\clGroup{\sqrt{n}}$ --- ideally covering on the way all clause vertices. To force an assignment-selector path to only cover clauses that are satisfied by the corresponding assignment, we rely on the fact that these paths must be isometric. To achieve this, we add \emph{distance modulator} gadgets that shorten by one the distance between assignments vertices and clauses that are not satisfied by this assignment. 

To ensure that distance modulators only act as a metric-changer and not as an alternative way for the assignement-selector paths, we use a number of cherries (i.e., induced paths of length~2 whose endpoints have degree one). 
We note that in our reduction, the number of leaves is exactly twice the number of allowed paths, forcing each solution path to join two leaves, which greatly facilitates the analysis of the reduction. 
The complete construction is illustrated in \Cref{fig:schematic}.

\myparagraph{Organization.}
The rest of the section is divided in three parts. 
In Section \ref{sec:IPP:eth-diam-tw:construction} we give the full description of the reduction, that we break into a subsection for the description of each gadget and the analysis of the solution size. 
Then, we give some useful properties and observations in Section \ref{sec:IPP:eth-diam-tw:properties}. 
We conclude with the proof of Theorem~\ref{thm:IPP:eth-diam-tw-hard} in Section \ref{sec:IPP:eth-diam-tw:proof}.

\subsection{Reduction}\label{sec:IPP:eth-diam-tw:construction}

We now describe the construction of our reduction and its different gadgets.

\subsubsection{Distance Modulators}
\label{sec:IPP:eth-diam-tw:construction:dist-modul}

In our construction, we use ``distance modulators'' to encode the formula. 
The main idea behind this distance modulator has been previously used 
in~\cite{DBLP:conf/icalp/FoucaudGK0IST24} under the name of 
``set representation gadget''. 
We first present the gadget in its general form and then specify how to adapt it for the reduction.

Consider two sets $A$ and $B$, an onto total function $\lambda\colon A \to B$, and an integer $q\ge 3$, we construct a \emph{distance modulator} gadget, denoted $D(A,B,\lambda,q)$ as follows.

Our objective is to construct a graph of small treewidth whose vertex set contains
vertices of $A$ and $B$, and for any $a \in A$ and $b \in B$, we have $a$ and $b$ at distance $q$ if $\lambda(a) = b$, and at distance $q - 1$ otherwise. 
Note that a naive way is to add paths of appropriate lengths between the vertices in $A$ and $B$.
However, the resulting graph would have treewidth $\Omega(|B|)$, assuming the size of $B$ is smaller than that of $A$. 

\begin{figure}
    \centering
    \includegraphics[page=2,width=0.5\linewidth]{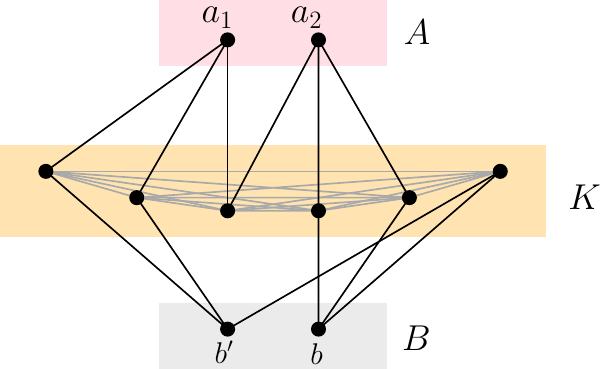}
     \caption{An example of a distance modulator. 
    Here, $\lambda(a_1)=b$, $\lambda(a_2)=b'$, and $q=3$.
    {Note that $a_1$ is at distance $3$ from $b$ whereas 
    it is at distance $2$ from $b'$. }}
    \label{fig:distance-modulator-main-description}
\end{figure}

Let $p$ be the smallest integer such that $|B| \le \binom{2p}{p}$ and $\calS_p$ be the collection of the $\binom{2p}{p}$
subsets of $[2p]$ that contain exactly $p$ integers. (We show later that having the property that $p=O(\log |B|)$ is enough for our purposes.) 
Then, we define $\setrep: B \to \calS_p$ as a one-to-one function by arbitrarily assigning a set in $\calS_p$ to a vertex in $B$. 
To construct $D(A,B,\lambda,q)$, we proceed as follows:

\begin{itemize}
    \item We start from the two (disjoint) vertex sets $A$ and $B$.

    \item We add a clique $K = \{u_1, u_2, \dots,  u_{2p}\}$ on $2p$ vertices, referred to as the \emph{central clique}. 

    \item For every $b \in B$ and for every $p' \in \setrep(b)$, we add the edge $(b, u_{p'})$.
    
    \item For each $a \in A$ 
    and each $p' \in [2p] \setminus \setrep(\lambda(a))$, we add a path of length $q -2$ from $a$ to $v_{p'} \in K$, that we call a \emph{connector path}. 
\end{itemize}

The above construction is illustrated in \Cref{fig:distance-modulator-main-description}.

\begin{lemma}
For each $a\in A$ and each $b\in B$, the distance between $a$ and $b$ in $D(A,B,\lambda,q)$ is 
 $q$ if $\lambda(a)=b$ and $q-1$ otherwise. Moreover, $D(A,B,\lambda,q)$ has treewidth $O(\log |B|)$.  
\end{lemma}

\begin{proof}
Let $a\in A$ and $b\in B$. First notice that the distance between $a$ and $b$ is either $q$ or $q-1$. 
First, suppose that $\lambda(a)\neq b$. Since $\setrep(\lambda(a))\neq \setrep(b)$ are distinct subsets of $[2p]$ of size $p$, there exists an element $p'$ contained in both $\setrep(b)$ and $[2p]\setminus \setrep(\lambda(a))$. Thus, there is a path of length $(q-2)+1$ from $a$ to $b$ that passes through $p'$. 

Suppose now that there is a path of length $q-1$
between $a$ and $b$. This path consists in a connector path from $a$ for some vertex $v_{p'}\in K$, with $p'\in [2p]$, and the edge $(v_{p'},b)$. This implies that $p'\in \setrep(b)$ and $p'\notin \setrep(\lambda(a))$, and further that $\lambda(a)\neq b$. 

We now study the treewidth of the gadget. Let $q' = \lceil 2 \cdot \log_2(|B|) \rceil$.
This implies that 
\[
    |B| \le 2^{q'} = \frac{4^{q'}}{2^{q'}} \le \frac{4^{q'}}{\sqrt{\pi \cdot {q'}}} \sim \binom{2{q'}}{{q'}}.
\]
The last step follows from the asymptotic estimate of the 
central binomial coefficient which states $\smash{\binom{2p}{p}\sim \frac{4^p}{\sqrt{\pi \cdot p}}}$~\cite{Sperner}. 
Hence, fixing a value of $p$ such that $p={O(\log (|B|)}$ suffices 
for our purpose. 

Now, note that the graph obtained from $\distanceMod{A}{B}{\lambda}{q}$ by deleting all vertices of $K$ is a collection of disjoint stars or subdivisions of stars.
Hence, the treewidth of $\distanceMod{A}{B}{\lambda}{q}$ is $O(|K|) = O(\log(|B|))$.
This concludes the proof. 
\end{proof}

In the reduction, $A$ will correspond to the set of all possible assignments of the variables of some part $V_i$ of the variables, and will hence have size $2^{O(\sqrt{n})}$.
The set $B$ will correspond to the clauses of some part $C_j$ of the clauses, and will hence have size $O(\sqrt{n})$.
It follows from our assumptions on the formula that every assignment of $V_i$ satisfies at most one clause in $C_j$.
We critically use this property of \textsc{Sparse 3-SAT}.
We will consider the following definition of $\lambda$: 
For $a \in A$, we set $\lambda(a) = b$ if and only if the assignment
corresponding to $a$ satisfies the clause corresponding to $b$.
Note however that such a function may not be defined for the full set $A$, since there may be assignments in $A$ that do not satisfy any clause. 
Thus, we will have a special element $b_0$ in $B$ such that
$\lambda(a) = b_0$ for such assignments $a$ that do not satisfy any clause.

\subsubsection{Encoding Assignments}
\label{sec:IPP:eth-diam-tw:construction:assignments}

For each variable group $V_i$, we construct an \emph{assignment-selector gadget} as follows. 
Let $\cycle{i}$ be the graph obtained by taking a cycle on $2^{\left|V_i\right|}+1$ vertices and attaching a (distinct) leaf to each vertex of the cycle.
Let $T_i$ be the set of degree~3 vertices of $\cycle{i}$.
See \Cref{fig:schematic} for an illustration. 
We arbitrarily choose $p_i\in T_i$ (on the figure it is the leftmost vertex) and define a bijection $\alpha\colon T_i \setminus \{p_i\} \to 2^{V_i}$.
This function will represent an assignment of the variables in $V_i$.
We extend this function by defining $\alpha(p_i)=\bot$.
Let $\mathcal{A}$ denote the disjoint union of the assignment-selector gadgets $A_1,\ldots,A_{\sqrt{n}}$.
We call $\mathcal{A}$ the \emph{assignment-selector} gadget.

In the following, we will assume the variables of $V_i$ to be labeled $\varV{i}{1},\ldots \varV{i}{|V_i|}$.

\subsubsection{Encoding Clauses}
\label{sec:IPP:eth-diam-tw:construction:clauses}

Recall that $C_j$ is a set of clauses, for any $j\in [\sqrt{n}]$.
For each such group of clauses, we consider an arbitrary labeling $\clGroupName{j}{1},\ldots,\clGroupName{j}{|C_j|}$ of its clauses; we shall refer to $\clGroupName{j}{\ell}$ as the $\ell^{th}$ clause in $C_j$.

For each $j\in [\sqrt{n}]$, we create a set $\clGroup{j}$ of $\sqrt{n}$ vertices as follows. 
For each $\ell\in [|C_j|]$, there is a \emph{clause vertex}
$\clGroupName{j}{\ell}\in \clGroup{j}$ representing the $\ell^{th}$ clause of $C_j$. 
For each $\ell\in [(\sqrt{n} - |C_j|)]$, $\clGroup{j}$ contains a dummy vertex $\cldummy{j}{\ell}$. 
Also we introduce {four} sets $\clGroup{-2}, \clGroup{-1}, \clGroup{0}$ 
and $\clGroup{\sqrt{n}+1}$, each with $\sqrt{n}$ vertices.
Let $\clGroup{-2}=\{y_1,\ldots,y_{\sqrt{n}}\}$ and 
$\clGroup{\sqrt{n}+1} = \{z_1,\ldots,z_{\sqrt{n}}\}$. 
Now, for each $j\in [1,\sqrt{n}]$, we add all possible edges between $\clGroup{j}$ and $\clGroup{j-1}$. 
Then, we add a matching between $\clGroup{-2}$ and $\clGroup{-1}$, $\clGroup{-1}$ and $\clGroup{0}$, and $\clGroup{\sqrt{n}}$ and $\clGroup{\sqrt{n}+1}$. There are no other edges in $\mathcal{C}$ other than 
the ones mentioned above. 
Let $\mathcal{C}$ denote the resulting graph, which is called the \emph{clause gadget}. 
See \Cref{fig:schematic} for an illustration. 

Note that for each $j\in [0,\sqrt{n}]$, the vertices of 
$\clGroup{j}$ and $\clGroup{j-1}$ induce a complete bipartite graph. 
Throughout the construction process, the vertices of 
$\clGroup{\sqrt{n}+1}$ will continue to have degree~1. 
We note that the treewidth and diameter of $\mathcal{C}$ are 
at most $2\sqrt{n}$ and $\sqrt{n}+2$, respectively.

\subsubsection{Connecting Gadgets}
\label{sec:IPP:eth-diam-tw:construction:connect}

Now we introduce more edges and vertices to connect the assignment-selector and clause gadgets constructed above. 

\begin{itemize}
    \item For each $i\in [\sqrt{n}]$, we add an edge between $y_i\in \clGroup{-2}$ and all the vertices in $T_i$, i.e., the vertices of degree~3 in $A_i$.
 
    \item Let $A$ be the set of degree~3 vertices in $\mathcal{A}$.
    Recall that $\phi$ is an instance of \textsc{Sparse 3-Sat} and hence any assignment of variables in $V_i$ satisfies at most one clause in $C_j$.
    If $|C_j| = \sqrt{n}$, then for an  
    instance of \textsc{Sparse 3-Sat},
    every such assignment should satisfy exactly one clause in $C_j$.
    If $|C_j| < \sqrt{n}$ and there is an assignment of variables in $V_i$ that does not satisfy any clause in $C_j$, we choose a dummy clause in $C_j$ and call it its corresponding clause.
    For each $j\in [\sqrt{n}]$, we define $\lambda_j\colon A\to C_j$ as follows. 
    Given a vertex $a\in T_i$ for some $i\in [\sqrt{n}]$, we define $\lambda_j(a)=c$ if $c\in C_j$ is satisfied by assigning the variables in $\alpha(a)$ to \true\ and the variables in $V_i\setminus\alpha(a)$ to \false, if such a clause exists, otherwise we map $a$ to a dummy clause.
    Note that $\lambda_j$ is a well-defined function, and it is total thanks to the dummy clauses vertices.
    Let {$B_j$} 
    denote the clause vertices of $\clGroup{j}$.
    Introduce a distance modulator $D_j=\distanceMod{A}{{B_j}}{\lambda_j}{{j+3}}$. 
    Let $\mathcal{P}$ be the set of all connector paths introduced over all $j\in [\sqrt{n}]$.

    \item For each $i\in [\sqrt{n}]$, technically, $D_i$ contains a copy, say $A'$, of the vertex set $A$. 
    In this step, we identify the vertices of $A'$ with $A$. 
    Similarly, $D_i$ contains a copy, say $B'_i$, of the clause vertices in $B_i$. 
    We also identify $B'_i$ with $B_i$. 
    Let $K_i$ denote the central clique in $D_i$ and 
    \[
        \mathcal{K}=\displaystyle\bigcup\limits_{i=1}^{\sqrt{n}} K_i.
    \]
    
    \item For each connector path $P$ with length at least~2 that connects two vertices $w,w'$ with $ w\in A, w'\in \mathcal{K}$, we do the following:
    \begin{itemize}
        \item we introduce a new leaf $t$ 
      adjacent to the neighbor of $w$ in $P$; and 
        \item we introduce a new leaf $t'$ 
      adjacent to the neighbor of $w'$ in $P$. 
     \end{itemize}
    
    \item For each $i\in [\sqrt{n}]$ and each vertex $w\in K_i$, create two new 
    vertices $w_1,w_2$ and make them adjacent to $w$. 
    Throughout the construction process, $w_1,w_2$ will remain adjacent 
    only to $w$, and thus will have degree~1.
    To avoid introducing more notations, from this point onwards,
    we include the pendant vertices attached to the vertices of 
    $K_i$, in the set $K_i$. See \Cref{fig:distance-modulator} for an illustration.
\end{itemize}

\medskip \noindent This concludes the description of the reduction.

\subsubsection{Solution size and intuition behind the reduction}


The reduction sets 
\begin{equation}\label{eq:optimal-size}
    k  = \frac{1}{2}\left[\Bigl(\displaystyle\sum\limits_{i=1}^{\sqrt{n}}2^{|V_i|}\Bigr) 
+ \Bigl(\displaystyle\sum\limits_{i=1}^{\sqrt{n}} 2 \cdot |K_i|\Bigr) 
+ \Bigl(2 \cdot |\mathcal{P}|\Bigr) 
+ \Bigl(\sqrt{n}\Bigr)\right]
\end{equation}
and returns
$(G_{\phi}, k)$ as a reduced instance. 

We informally justify the value of $k$ along with the core idea of the reduction.
The graph is constructed such that any IP-partition has to include a number of paths equal to half of the second and third terms above, to partition vertices in the central cliques and in connector paths, respectively. 
This is done using leaves and cherries.
Then, recall that the assignment encoding gadget $\cycle{i}$ 
is a cycle on $2^{\left|V_i\right|}+1$ vertices with 
a unique leaf attached to each vertex of the cycle.
It is easily seen that any IP-partition uses half of the first term to cover
all but one vertex on each cycle, where these paths have four vertices each, and consist of two consecutive vertices on the cycle plus their respective leaves.
After this, there is a remaining budget of $\sqrt{n}$.
Note that there is one vertex plus its leaf neighbor in each $\cycle{i}$, as well as all the vertices in $C'_{\sqrt{n + 1}}$, that still need to be partitioned.
We will argue that the partition needs $\sqrt{n}$ isometric paths 
such that each of them starts in $\cycle{i}$, ends in a vertex in
$C'_{\sqrt{n}+1}$, and covers one vertex in each of $C'_{-2}, C'_{-1}, C'_{0}, C'_{1}, \dots,
C'_{\sqrt{n}}$ along the way.
The distance modular gadget is constructed to ensure that an isometric path starting from a vertex $a\in V(A_i)$ can contain a vertex $c_j\in C'_i$ if and only if the assignment corresponding to $a$ satisfies $c_j$. 

\begin{figure}[ht]
    \centering
    \includegraphics[page=3,width=0.5\linewidth]{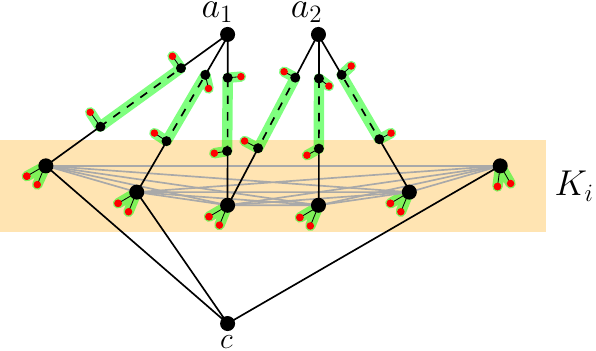}
     \caption{$K_i$  is the central clique, and $c$ is a clause in $C_i$. Let $c=(x_1\lor \overline{x_2} \lor x_3)$. Let $a_2$ and $a_1$ be two vertices of the assignment-selector gadget and $\alpha(a_2)=\{x_2\}$ and $\alpha(a_1)=\{x_1\}$. Since setting $x_1=\true$ satisfies $c$, $a_1$ is not connected to the neighborhood of $c$ in $K_i$. The dashed lines indicate connector paths. The paths from $a_i$'s to $K_i$ are called connector paths. Their length depends on which of the sets $C'_j$ it is attached to.  A number of cherries (red vertices) is added, which forces only one relevant way to partition the gadget (green paths). }
    \label{fig:distance-modulator}
\end{figure}

\subsection{Properties}
\label{sec:IPP:eth-diam-tw:properties}

Note that the vertex set of $G_{\phi}$ consists of the assignment 
gadget $\mathcal{A}$, the clause gadget 
$\mathcal{C}$, the vertices in $\mathcal{K}$, the leaves attached to the vertices in $\mathcal{K}$, the connector paths in $\mathcal{P}$, and the leaves attached to these connector paths. 
The total number of leaves in $G_{\phi}$ is $2k$, where $k$ is the value set in \Cref{eq:optimal-size}.
We prove the following two claims, that facilitate the proof of 
correctness of the reduction.

\begin{lemma}\label{lem:connector-path-cover}
 Let $Q$ be a connector path between two vertices $w,w'$ of $G_{\phi}$. Let $t,t'$ be the two 
 leaves of $G_{\phi}$ such that $t$ is adjacent to the neighbor of $w$ in $Q$, and $t'$ is adjacent 
 to the neighbor of $w'$ in $Q$. Then, the path $Q'$ induced by $\left(V(Q)\setminus \{w,w'\}\right)\cup \{t,t'\}$ is an isometric path. 
\end{lemma}

\begin{proof}
Let $w\in V(A_i)$ and $w'\in K_j$ for some $i,j\in [\sqrt{n}]$. Then, the length of $Q'$ is $j+1$. 
Consider an induced path $Q'_1$ between $t$ and $t'$ which is distinct from $Q'$. 
Clearly, $Q'_1$ must contain $w,w'$; let $Q'_2$ be the subpath of $Q'_1$ between $w$ and $w'$. Observe that $|E(Q'_1)|=|E(Q'_2)|+4$. 
If $Q'_2$ does not contain any vertex from a connector path (other than $w,w'$), 
then the length of $Q'_2$ is at least $j+1$, and therefore, the length of $Q'_1$ is at least $j+5$. 
Suppose $Q'_2$ contains vertices from another connector path between $v\in V(A_k)$ and $v'\in K_{k'}$ for some $k,k'\in [\sqrt{n}]$. 
In this case, the length of $Q'_2$ is at least $k'+ |j-k'|\geq j$, and therefore the length of $Q'_1$ is at least $j+4$. 
Hence, $Q'$ is an isometric path between its endpoints.
\end{proof}

\begin{lemma}\label{lem:dist-assign}
    For some $i\in [\sqrt{n}]$, let $v_i$ be a vertex of degree~3 of $A_i$. For an integer $j\in [\sqrt{n}]$, 
    let $c$ be a clause vertex in $\clGroup{j}$. If assigning the variables in $\alpha(v_i)$ to \true\ and the 
    variables in $V_i\setminus \alpha(v_i)$ to \false\ does not satisfy $c$, then the distance between $v_i$ and $c$ is at 
    most $j+2$ in $G_{\phi}$. Otherwise, the distance between $v_i$ and $c$ is $j+3$.
\end{lemma}
\begin{proof}
When assigning the variables in $\alpha(v_i)$ to \true\ and the variables in $V_i\setminus \alpha(v_i)$ to 
\false\ does not satisfy $c$, observe that there exists a $w'\in \fnb{c}$ such that there is a connector path 
$P$ between $v_i$ and $w'$. By construction, the length of $P$ is $j+1$, and $w'$ is adjacent to $c$. Hence, the distance 
between $v_i$ and $c$ is at most $j+2$. 

Suppose $c$ is satisfied. Let $P$ be an induced path between $v_i$ and $c$.  
Since there is no connector path between $v_i$ and a vertex of $\fnb{c}$, if $P$ contains vertices from some connector path, 
the length of $P$ is at least $j+3$. Otherwise, $P$ contains exactly one vertex from $\clGroup{j}$ for each {$j\in [-2,j]$}. 
Hence, the length of $P$ is $j+3$. 
\end{proof}

\subsection{Proof of \Cref{thm:IPP:eth-diam-tw-hard}}
\label{sec:IPP:eth-diam-tw:proof}

We are now ready to prove Theorem~\ref{thm:IPP:eth-diam-tw-hard}, that we split into three lemmas.
Let us recall that the reductions sets
\[
    k  = \frac{1}{2}\left[\Bigl(\displaystyle\sum\limits_{i=1}^{\sqrt{n}}2^{|V_i|}\Bigr) 
+ \Bigl(\displaystyle\sum\limits_{i=1}^{\sqrt{n}} 2 \cdot |K_i|\Bigr) 
+ \Bigl(2 \cdot |\mathcal{P}|\Bigr) 
+ \Bigl(\sqrt{n}\Bigr)\right]
\]
given an instance of \textsc{Sparse 3-SAT} on $n$ variables.

\begin{lemma}
\label{lem:diam-tw-hard:forward}
    If $(\phi, \{V_1, \dots, V_{\sqrt{n}}\}, \{C_1, \dots, C_{\sqrt{n}}\})$ is a \yes-instance of \textsc{Sparse 3-SAT}, then $(G_{\phi}, k)$ is a \yes-instance of \SoPP.
\end{lemma}

\begin{proof}
In this section, we show that if $\phi$ has a satisfying assignment, then $G_{\phi}$ has an IP-partition of 
cardinality $k $. Fix a satisfying assignment $f\colon V\rightarrow \{\true,\false\}$. 
For each $i\in [\sqrt{n}]$, let $v_i\in V(A_i)$ be {the} vertex such that $\alpha(v_i)=\{x\in V_i : f(x)=\true\}$. 
In other words, the variables in $\alpha(v_i)$ are exactly the variables in $V_i$ that were assigned to \true. 
We construct an IP-partition $\mathcal{R}$ of $G_{\phi}$ of cardinality $k$ as follows. 

\begin{itemize}
    \item Initialize $\mathcal{R}=\emptyset$.
    For each vertex $w\in \mathcal{K}$ that is not a leaf, let $w_1$ and $w_2$ be the two leaves adjacent to $w$. 
    We add the cherry $(w_1,w,w_2)$ to $\mathcal{R}$.

    \item For each connector path $Q$, we do the following. Let $w$ and $w'$ be the endpoints of $Q$. 
    Let $t,t'$ be the two leaves of $G_{\phi}$ such that $t$ is adjacent to the neighbor of $w$ in $Q$, and $t'$ is adjacent to the neighbor of $w'$ in $Q$. 
    Let $P_{ww'}$ denote the path induced by $(V(Q)\setminus \{w,w'\})\cup \{t,t'\}$. 
    By \Cref{lem:connector-path-cover}, $P_{ww'}$ is an isometric path. 
    We add it to $\mathcal{R}$.

    \item For each $i\in [\sqrt{n}]$, we do the following. 
    Recall that $A_i$ consists of a cycle $Q_i$ of order $2^{|V_i|}+1$ with one leaf attached to each vertex of $Q_i$. 
    Consider the path $P_i$ induced by $V(Q_i)\setminus \{v_i\}$. 
    Observe that $P_i$ consists of $2^{|V_i|}$ vertices, each with a unique leaf attached. 
    The vertices of $P_i$ along with the attached leaves can be partitioned using $2^{|V_i|-1}$ many isometric paths of four vertices each, and we put such an IP-partition in $\mathcal{R}$.

    \item For each $i\in [\sqrt{n}]$, we do the following. 
    Let $v'_i$ be the leaf adjacent to $v_i$. 
    Let $u_0$ be a vertex of $\clGroup{0}$ which is not in any path of $\mathcal{R}$ yet. 
    For $j\in [\sqrt{n}]$, define $u_j$ as follows. 
    
    \begin{itemize}
        \item If a variable of $V_i$ appears in some clause of $C_j$, then there exists exactly one clause vertex $c\in \clGroup{j}$ 
        that contains a variable in $V_i$. Moreover, $c$ is satisfied by $\alpha(v_i)$ by the choice of $f$ and $\alpha$. 
        If $c$ is not covered by any path in $\mathcal{R}$, 
        then define $u_j=c$.
    
        \item Otherwise, define $u_j$ to be a dummy vertex of $\clGroup{j}$ which does not appear in any path of $\mathcal{R}$.
    \end{itemize}
    Finally, let $u'_i$ be the vertex in $\clGroup{\sqrt{n}+1}$ which is adjacent to $u_{\sqrt{n}}$. Recall that $y_i\in \clGroup{-2}$ 
    is adjacent to $v_i$. Now, construct a path  $Q$ induced by $v'_i,v_i,y_i,u_0,\ldots,u_j,\dots,z_i$. We put $Q$ in $\mathcal{R}$. 
\end{itemize}

    We now prove that the set $\mathcal{R}$ constructed above 
    is an IP-partition of $G_{\phi}$ of cardinality $k$.
    Clearly, $\mathcal{R}$ is a partition of the vertex set, and the endpoints of the paths in $\mathcal{R}$ are leaves 
    of $G_{\phi}$. Hence, $|\mathcal{R}|=k$. 
    The paths added in the first three bullets are clearly isometric paths. 
    Let $P\in \mathcal{R}$ be a path as described in the fourth (and last) bullet. 
    By construction, one endpoint of $P$ is $v'_i$, which is the leaf adjacent to $v_i\in V(A_i)$ for some $i\in \sqrt{n}$. 
    Also, $v_i\in V(P)$. 
    We argue that for each vertex $c\in V(P)$, the subpath of $P$ between $v_i$ and $c$ is isometric. 
    
    Clearly, if $c\in \clGroup{-2} \cup \clGroup{-1} \cup \clGroup{0}$, the statement is true. 
    Let $c\in \clGroup{j}$ for some $j\in [\sqrt{n}]$. 
    The distance between $v_i$ and $c$ in the path $P$ is {$j+3$}.
    If $c$ is a clause vertex, then by definition, assigning $\alpha(v_i)$ to \true{} and $V_i\setminus \alpha(v_i)$ to \false{} satisfies $c$. 
    \Cref{lem:dist-assign} implies that the subpath between $v_i$ and $c$ is an isometric path. 
        
    Using a similar argument as in \Cref{lem:dist-assign},
    we can prove the following:
    For some $i\in [\sqrt{n}]$, let $w'$ be a vertex of degree~3 of $A_i$. 
    For an integer $j\in [\sqrt{n}]$, let $d$ be a dummy vertex in $\clGroup{j}$. 
    Then, the distance between $w'$ and $d$ is {$j+3$}.
    As $w$ is a dummy vertex, this implies the claim. 
    Finally, if $c\in \clGroup{\sqrt{n}+1}$, then consider the vertex $c'\in V(P)$ adjacent to $c$. 
    Since $c$ is a leaf, and $c'\in \clGroup{\sqrt{n}}$, the above arguments complete the proof of the lemma.
\end{proof}

\begin{lemma}
\label{lem:diam-tw-hard:backward}
    If $(G_{\phi}, k)$ is a \yes-instance of \SoPP,
    then $(\phi, \{V_1,\allowbreak \dots, V_{\sqrt{n}}\}, \{C_1, \dots, C_{\sqrt{n}}\})$
    is a \yes-instance of \textsc{Sparse 3-SAT}.
\end{lemma}

\begin{proof}
We show that if $G_{\phi}$ has an IP-partition of cardinality 
$k$ then $\phi$ is satisfiable. Let $\mathcal{R}$ be an IP-partition of $G_{\phi}$ of cardinality $k$. \textcolor{black}{Since the number of leaves is $2k$, and a path contains at most two leaves, each path in $\mathcal{R}$ must go from one leaf to another. We now explain why these endpoints must be connected as in the previous lemma.} 

\begin{enumerate}[label=(\alph*)]
    \item\label{it:a} First, note that vertices $w\in \mathcal{K}$ have two attached pendant leaves which together with these leaves, induce a cherry. Thus, by \Cref{lem:cherry-lemma} we may assume that these cherries belong to~$\mathcal{R}$.

    \item\label{it:b} Let $Q$ be a connector path between two vertices $w,w'$ of $G_{\phi}$, with $w'$ being in $\mathcal{K}$. 
    Let $t,t'$ be the two leaves of $G_{\phi}$ such that $t$ is adjacent to the neighbor of $w$ in $Q$, and $t'$ 
    is adjacent to the neighbor of $w'$ in $Q$. 
    The path $P\in \mathcal{R}$ that contains $t'$ must connect to another leaf, and, in particular, it must pass through $w$'s neighbor in the connector path. If $t\notin P$, then $t$ must be covered by a path that contains only one leaf,  which is a contradiction. Thus, there is a path in $\mathcal{R}$ that connects $t$ to $t'$ and contains all the vertices of $Q$ between $w$ and $w'$. By \Cref{lem:connector-path-cover}, this is indeed an isometric path. 
    
    \item\label{it:c} 
    For each $i\in [\sqrt{n}]$, let us argue that there is exactly one isometric path in $\mathcal{R}$ having one endpoint in $A_i$ and the other endpoint not in $A_i$. First, there is at least one such path since the number of leaves in $A_i$ is odd. Then, there is at most one such path, since two such paths 
    would contain the same vertex $y_i\in \clGroup{-2}$, which would be a contradiction with disjointness. 
    
    \item\label{it:d} Let $A$ be a connected component of the assignment gadget $\mathcal{A}$ and $P\in \mathcal{R}$ be a path with exactly one endpoint in $A$, that we call $u$.  
    We argue that its other endpoint $v$ must lie in $\clGroup{\sqrt{n}+1}$. 
    Due to \ref{it:a} and \ref{it:b}, we have that $v$ must either lie in the assignment gadget, or in $\clGroup{\sqrt{n}+1}$. 
    Assume that there exists a connected component $A'\not= A$ that contains $v$. Since the distance between these vertices is $10$, the length of $P$ is at most $10$. 
    On one hand, $\mathcal{C}$ has exactly $n+4\sqrt{n}$ vertices and by \ref{it:c} these vertices are covered by at most $\sqrt{n}$ many isometric paths of $\mathcal{R}$. 
    On the other hand, the diameter of $\mathcal{C}$ is $3+\sqrt{n}$ and thus an isometric path covers at most $4+\sqrt{n}$ vertices of $\mathcal{C}$. 
    This implies that each path must cover exactly $\sqrt{n}+4$ of $\mathcal{C}$. For sufficiently large values of $n$, we have $\sqrt{n}+4> 10$, which shows that $v$ cannot be in $\mathcal{A}$. 
\end{enumerate}

The sought assignment of the variables is encoded by the paths from $\mathcal{A}$ for $C'_{\sqrt{n}+1}$. For each $i\in [\sqrt{n}]$, let $w_i$ be the vertex of degree~3 covered by one of these paths. 
Then, assign all variables of $\alpha(w_i)$ to \true{} and all variables of $V_i\setminus \alpha(w_i)$ to \false. 
We show that the above assignment satisfies $\phi$. 

Consider a clause vertex $c\in \clGroup{j}$ for some $j\in [\sqrt{n}]$. 
Let $P\in \mathcal{R}$ be an isometric path that contains $c$. 
One endpoint of $P$ lies in a connected component $A_i$ for some $i\in [\sqrt{n}]$, and 
let $w_i\in V(A_i)$ be the vertex of $P$ which has degree three in $A_i$. Observe that the distance between $w_i$ and $c$ is {$j+3$}.
Thus, by \Cref{lem:dist-assign}, $c$ is satisfied by the assignment. 
\end{proof}

\begin{lemma}\label{lem:diam-tw-hard:parameters}
    The graph $G_\phi$ has pathwidth $O\left(\sqrt{n}\log n\right)$ and diameter $O\left(\sqrt{n}\right)$.
\end{lemma}

\begin{proof}
    First note that the graph induced by the assignment gadget, together with $\clGroup{-2}, \clGroup{-1}$, $\clGroup{0}$, and the distance modulators, has pathwidth $O(\sqrt{n}\log n)$. 
    Similarly, the graph induced by the clause gadget and the distance modulators has pathwidth $O\left(\sqrt{n}\log n\right)$. 
    Since all edges of $G_{\phi}$ are covered by the above subgraphs, the treewidth of $G_{\phi}$ is $O\left(\sqrt{n}\log n\right)$. 

    To show that the diameter of $G_{\phi}$ is $O\left(\sqrt{n}\right)$, first observe that the graph induced by the 
    clause gadget and the distance modulators has diameter $O(\sqrt{n})$. Then, observe that the graph induced by the assignment gadget, together with $\clGroup{-2}, \clGroup{-1}$ and $\clGroup{0}$, has diameter at most 10.
    Finally, observe that the distance between any vertex of the assignment gadget and any vertex of a distance modulator is at most $O(\sqrt{n})$. 
    Combining the above arguments, we have that the diameter of $G_{\phi}$ is $O\left(\sqrt{n}\right)$.    
\end{proof}

Note that given an instance $\phi$ of \textsc{Sparse 3-SAT}, the graph $G_{\phi}$ can be constructed in $2^{O(\sqrt{n})}$ time. 
By Lemmas~\ref{lem:diam-tw-hard:forward} and \ref{lem:diam-tw-hard:backward}, and \ref{lem:diam-tw-hard:parameters}, if there is an algorithm for \SoPP running in time 
$\diam^{o(\pw^2/(\log^3\pw))} \cdot |V(G_{\phi})|^{O(1)}$, then \textsc{Sparse 3-SAT} will admit a $2^{o(n)}$ algorithm, 
which contradicts \Cref{prop:sparse-hard}.
We conclude with Theorem~\ref{thm:IPP:eth-diam-tw-hard}, that we restate here.

\restateIPPdiamtwhard*

\subsection{A Note about Kernelization}\label{sec:kernel}

We conclude this section by another negative result on the parameterization by diameter and pathwidth.

\begin{proposition}
    \SoPP{} does not admit a polynomial kernel when parameterized by $\diam + \pw$, unless
$\NP \subseteq \coNP/\poly$.
\end{proposition}

\begin{proof}
	This follows from a simple \textsc{And}-cross-composition.
	See \cite[Chapter 15.1.3]{cygan2015parameterized} for formal
	definitions and precise statements.
	Consider the following equivalence relation defined on
	instances of \SoPP{}.
	Two instances $(G_i, k_i)$ and $(G_j, k_j)$ are in the same equivalence class
	if and only if $k_i = k_j$ and the number of vertices in $G_i$ is the same as the 
	number of vertices in $G_j$.
	It is easy to verify that this is a polynomial equivalence relation 
	(See \cite[Definition 15.7]{cygan2015parameterized}).
	Consider the following \textsc{And}-composition that takes $t$ many instances
	$(G_1, k), (G_2, k), \cdots, (G_t, k)$ of \SoPP{}
	that are in the same equivalence class, 
	and returns another instance $(G', k')$ of \SoPP:
	To construct $G'$, the reduction starts with a disjoint union of $G_1, G_2, \dots, G_t$.
	It then adds a path $(u, v, w)$ and makes $v$ adjacent 
	with an arbitrary vertex in every connected component in $G_i$ for every $i \in [t]$.
	It sets $k' = k \cdot t + 1$ and returns the instance $(G', k')$.
	The forward direction of the reduction follows from the fact that
	combining solutions for individual instances along with
	$(u, v, w)$ constructs a solution for $(G', k')$.
	In the reverse direction, by Lemma~\ref{lem:cherry-lemma} we can assume without loss of generality that an optimal path partition of $G'$ contains the cherry $(u, v, w)$.
	This implies that $(G', k')$ is a \yes-instance of \SoPP{}
	if and only if $(G_i, k)$ is a \yes-instance of \SoPP{}
	for every $i \in [k]$.
	It is easy to verify that $\diam(G') + \pw(G')$ are upper-bounded 
	by four times the maximum number of vertices in $G_i$ for any $i \in [t]$.
	Hence, \cite[Theorem 15.12]{cygan2015parameterized} 
	implies that \SoPP{} does not admit 
	a polynomial kernel unless $\NP \subseteq \coNP/\poly$.
\end{proof}

\section{Conclusion}\label{sec:conclusion}

In this article, we studied the parameterized complexity of \SoPP. 
We proved that the problem admits an \XP-algorithm but is \W[1]-hard when
parameterized by the treewidth (and pathwidth) of the input graph.
This improves the existing results
from~\cite{dumas2024graphs} and \cite{PPP}, and
answers open questions mentioned in these articles. In addition, we obtained an \FPT\ algorithm (with running time $\diam^{O(\tw^2)}\cdot n^{O(1)}$) parameterized by diameter and
treewidth, and proved a conditional randomized \ETH-based lower bound
that differs from the algorithm running time only by a poly-logarithmic factor. As noted in Section~\ref{sec:intro}, this type of running time is relatively rare in the literature. Our result shows that \SoPP behaves more like other metric-based problems with respect to treewidth and diameter, as opposed to its non-metric counterpart (\textsc{Path Partition}, see~\cite{PCcaldam}).

As a future research question, we wonder 
whether the problem is \FPT\ or \W[1]-hard when parameterized by solution size $k$. 
Recall that there is an algorithm running in time
$f(k)\cdot n^{O(k)}$~\cite{dumas2024graphs}, and that the treewidth is upper-bounded by a function of $k$ in any \textsc{Yes}-instance~\cite{dumas2024graphs}. Note that the analogous problem is known to be \W[1]-hard for $k$ on DAGs~\cite{PPP}, but the authors reported being unable to prove the analogous result for undirected graphs.

Another interesting question is whether \SoPP becomes \FPT\ for treewidth on planar graphs. 
We note that this question has been raised for other metric-based problems such as \textsc{Metric Dimension}~\cite{DBLP:journals/algorithmica/BonnetP21}.

\paragraph*{Acknowledgments.}
The authors would like to thank the organizers and participants of the 2024 Workshop of the AGC team of LIMOS at Saint-Jacques-d'Ambur for the relaxed yet inspiring atmosphere, where this research was initiated.

Florent Foucaud was supported by the ANR project GRALMECO (ANR-21-CE48-0004), the French government IDEX-ISITE initiative 16-IDEX-0001 (CAP 20-25), the International Research Center ``Innovation Transportation and Production Systems'' of the I-SITE CAP 20-25, and the CNRS IRL ReLaX.

\bibliographystyle{alpha}
\bibliography{references} 

\end{document}